\documentclass[journal, 11pt, onecolumn, onecolumn]{IEEEtran}

\ifCLASSINFOpdf
\else
\fi

\usepackage{amssymb,amsbsy,graphicx,amsmath,amsthm,amscd}
\usepackage[margin=1in]{geometry}
\usepackage{url}

\usepackage[noadjust]{cite}

\usepackage{mathptmx}       % selects Times Roman as basic font
\usepackage{helvet}         % selects Helvetica as sans-serif font
\usepackage{courier}        % selects Courier as typewriter font
\usepackage{type1cm}        % activate if the above 3 fonts are
\usepackage{graphicx}     

\usepackage{mathtools, nccmath}

\usepackage{verbatim}                       
\usepackage[utf8]{inputenc}
\usepackage{hyperref}
\usepackage{theoremref}
\usepackage{bbm}
\usepackage{xcolor}
\usepackage{algorithm,algpseudocode}%http://ctan.org/pkg/{algorithms,algorithmx}

\allowdisplaybreaks

\newtheorem{thm}{Theorem}

\newtheorem{lem}[thm]{Lemma}
\newtheorem{prop}[thm]{Proposition}
\newtheorem{cor}[thm]{Corollary}

\newtheorem{dfn}[thm]{Definition}

\newtheorem{rem}{Remark}

\newcommand{\ii}{\mathbbm{1}}

\newcommand{\R}{\mathbb R}

\newcommand{\inner}[1]{\left\langle #1\right\rangle}
\newcommand\numberthis{\addtocounter{equation}{1}\tag{\theequation}}

\theoremstyle{definition}

\algnewcommand{\Inputs}{%
	\State \textbf{Inputs:}
}
\algnewcommand{\Initialize}{%
	\State \textbf{Initialize:}
}
\algnewcommand{\Outputs}{%
	\State \textbf{Outputs:}
}
\algnewcommand{\ForLoop}{%
	\State \textbf{For $k=0,1,2,...$ do}
}
\algnewcommand{\ForEnd}{%
	\State \textbf{End For}
}
\algnewcommand{\Iterate}{%
	\State \textbf{Iterate:}
}

\begin{document}

%\title{\huge NBIHT:  An efficient algorithm for $1$-bit compressed sensing that achieves optimal error dependence on number of measurements}
%\title{\huge NBIHT:  An Efficient Algorithm for $1$-bit Compressed Sensing with Optimal Sample Complexity}
\title{\huge NBIHT:  An Efficient Algorithm for $1$-bit Compressed Sensing with Optimal Error Decay Rate}

\author {Michael P. Friedlander,
\and Halyun Jeong,
\and Yaniv Plan,
\and and {\"O}zg{\"u}r Y{\i}lmaz
\thanks{M. P. Friedlander is with the Department of Computer Science and the Department of Mathematics,  The University of British Columbia, Vancouver, BC, Canada (e-mail:mpf@cs.ubc.ca).

H. Jeong. Y. Plan, and {\"O} Y{\i}lmaz are with the Department of Mathematics,  The University of British Columbia, Vancouver, BC, Canada (email: hajeong@math.ubc.ca; yaniv@math.ubc.ca; oyilmaz@math.ubc.ca).

M. P. Friedlander was partially supported by the Office of Naval Research through grant N00014-17-1-2009.
Y.~Plan was partially supported by an NSERC Discovery Grant (22R23068), a PIMS CRG 33: High-Dimensional Data Analysis, and a Tier II Canada Research Chair in Data Science. {\"O}.~Y{\i}lmaz was partially supported by an NSERC Discovery Grant (22R82411) and PIMS CRG 33: High-Dimensional Data Analysis. H.~Jeong was funded in part by the University of British Columbia Data Science Institute (UBC DSI) and by the Pacific Institute of Mathematical Sciences (PIMS). 

}}

\maketitle

\begin{abstract}
The Binary Iterative Hard Thresholding (BIHT) algorithm is a popular reconstruction method for one-bit compressed sensing due to its simplicity and fast empirical convergence. There have been several works about BIHT but a theoretical understanding of the corresponding approximation error and convergence rate still remains open. 

This paper shows that the normalized version of BIHT (NBHIT) achieves an approximation error rate optimal up to logarithmic factors. More precisely, using $m$ one-bit measurements of an $s$-sparse vector $x$, we prove that the approximation error of NBIHT is of order $O \left(1 \over m \right)$ up to logarithmic factors, which matches the information-theoretic lower bound $\Omega \left(1 \over m \right)$ proved by Jacques, Laska, Boufounos, and Baraniuk in 2013. To our knowledge, this is the first theoretical analysis of a BIHT-type algorithm that explains the optimal rate of error decay empirically observed in the literature. This also makes NBIHT the first provable computationally-efficient one-bit compressed sensing algorithm that breaks the inverse square root error decay rate $O \left(1 \over m^{1/2} \right)$.%\cite{jacques2013robust, boufounos2015quantization, liu2019one,jacques2013quantized}. 

\end{abstract}


% \fi
\IEEEpeerreviewmaketitle

\section{Introduction}
%\IEEEPARstart{C}{ompressed} sensing is an efficient
Compressed sensing is an efficient signal acquisition and recovery paradigm that has received a considerable amount of attention from both researchers and practitioners. This novel paradigm has led to satisfactory solutions to many interesting problems in modern signal processing and machine learning that were considered to be intractable by traditional approaches. In the most basic setup of compressed sensing, one aims to recover an unknown $s$-sparse signal $x \in \R^N$, i.e., a vector with at most $s$ nonzero entries, from its linear measurements $Ax$ where $A$ is a known $m \times N$ matrix with $m \ll N$.

For measurement matrices with certain conditions, e.g., the \emph{restricted isometry property (RIP)}, which is satisfied for a large class of matrices including Gaussian random matrices (random matrices whose entries are i.i.d. standard Gaussian random variables), the typical results in compressed sensing guarantee that signal recovery is possible if the number of measurements $m \gtrsim s \log (2N/s)$ \cite{foucart2013invitation, eldar2012compressed}. 

On the other hand, because modern technology relies on digital computers, it is necessary to quantize measurements to store and process them. The discrete nature of  quantization presents exciting challenges in efficient signal recovery. One of the central challenges in this area is \emph{quantized compressed sensing}, the problem of recovering a sparse vector from its quantized measurements $Q(Ax)$ where $Q(\cdot)$ is a \emph{quantizer} that maps  vectors in $\R^m$ to elements of a finite set.
There are mainly two approaches to quantization in the setting of compressed sensing. 

The first approach is to use noise-shaping quantizers which aim to design vector quantizers such that the quantization error in the measurement space lies in directions away from the set of all possible unquantized measurements; thus, the quantization error can be mostly filtered out. Examples of noise-shaping quantizers in compressed sensing, which can be one-bit or multi-bit per measurement include $\Sigma\Delta$ quantization \cite{gunturk2013sobolev, saab2018quantization}, $\beta$-condensation \cite{chou2015noise,chou2016distributed}, and adaptive quatization \cite{baraniuk2017exponential}. These quantizers can reach approximation error rates decaying exponentially in $m$ directly \cite{deift2011optimal} or after additional encoding \cite{saab2017compressed}. This nearly optimal compression rate, however, comes at the expense of introducing memory components or sophisticated subroutine \cite{baraniuk2017exponential}  in the feedback quantization (encoding) process --- often not desired or sometimes not practical. 

The second approach to quantization is memoryless scalar quantization (MSQ), which quantizes each measurement separately. In other words, the quantizer $Q$ acts on $Ax$ elementwise, so it can be implemented without the need for analog memory. Because of its non-adaptive nature in the encoding process, MSQ is more suitable to parallel or distributed computing environment. One important case of quantized compressed sensing is when the quantizer only has two levels, i.e., it is one-bit. One-bit quantizers are easily implementable in hardware due to its simplicity, storage benefits, and robustness to scaling errors.

\emph{One-bit compressed sensing} was first introduced by Boufounos and Baraniuk \cite{boufounos20081}. The goal is to recover an $s$-sparse signal $x$ from its signed linear measurements, that is, from $b = \text{sign}(Ax)$ for a Gaussian random matrix $A$. Here $\text{sign}(\cdot)$ is applied elementwise on $Ax$ where the scalar function $\text{sign}(w)$ is equal to $1$ if $w > 0$ and $-1$ otherwise. This is a quantized compressed sensing problem where the $\text{sign}(\cdot)$ function is the one-bit memoryless scalar quantizer. Because one-bit measurements $\text{sign}(Ax)$ are scale invariant, one cannot recover the magnitude of $x$. Consequently, in the one-bit compressed sensing problem, we may assume without loss of generality that $x$ is a unit vector and restrict the signal set of interest to sparse vectors on the unit sphere. 

Another motivation for one-bit compressed sensing is the sparse binary response model in which we want to recover a sparse signal $x$ from binary observations  $\text{sign}(Ax + \epsilon)$ where $\epsilon$ is a noise vector \cite{ plan2016high, plan2016generalized, davenport20141, horowitz2009semiparametric}. Note that in these applications we do not have access to $Ax$ but only ${1 \over 2}(\text{sign}(Ax + \epsilon)+ 1)$, its binary labels, which excludes the vector quantization idea (here, the notation $1$ indicates a vector of all-ones).
Interestingly, the noiseless setting ($\epsilon = 0$) which is identical to the one-bit compressed sensing model, has been thought difficult to analyze since a certain amount of noise is required to make the traditional maximum likelihood analysis feasible \cite{plan2016generalized, davenport20141}. Our theory provides a near-optimal theoretical guarantee in the number of observations in this case.

Jacques et al. \cite{jacques2013robust} proposed \emph{binary iterative hard thresholding (BIHT)} as a reconstruction method for one-bit compressed sensing, which is our main subject of study. BIHT and its variants are state-of-the-art reconstruction algorithms for one-bit compressed sensing because of their low computational cost and superior empirical recovery accuracy in low-noise or noiseless settings \cite{jacques2013robust, boufounos2015quantization, liu2019one}. Formally, BIHT is a subgradient method with hard thresolding to the set of $s$-sprase vectors to solve the following nonconvex optimization problem \cite{jacques2013robust,liu2019one}:
\begin{equation} \label{prob:BIHT_optimization} 
\min_{z \in \R^N} \|Az\|_1 - \inner{b, Az} \quad
\text{subject to} \quad \|z\|_0 \le s, \; \|z\|_2 = 1.
\end{equation}

Other BIHT-type algorithms appeared in \cite{jacques2013robust} including {\it{normalized BIHT (NBIHT)}} with an addtional nonconvex projection step onto the unit sphere or the \emph{hinge loss function BIHT} by replacing the objective in \eqref{prob:BIHT_optimization} with the hinge-loss function.

There are various sources in the formulation \eqref{prob:BIHT_optimization} that make the convergence analysis of the BIHT challenging: First, the two constraints in \eqref{prob:BIHT_optimization} are nonconvex, so BIHT-type algorithms perform a nonconvex projection after the gradient step. Moreover, the objective function is not strictly convex and is not differentiable. As a consequence, the subgradients in BIHT-type algorithms have discontinuities, so traditional approaches based on the RIP are not directly applicable.

However, thanks to the randomness of measurement matrix $A$, we can use a multicale analysis to show that NBIHT converges to within a small radius of $x$; At large scales, i.e., for $z$ far from $x$, a RIP-type property holds. As $z$ approaches $x$, the RIP constants become larger until the near isometry breaks down when $z$ gets within the radius of the approximation error. This allows us to show that the iterates of NBIHT contract until they land within the radius of approximation error. We call this multiscale RIP-type guarantee the \emph{restricted approximate invertibility condition (RAIC)}, and believe it may be useful in other problems involving binary measurements. This forms the foundation of our analysis for NBHIT. 

Our main result shows that the NBIHT iterates $x_k$ obey 
\[
\|x_{k} - x\|_2 \le   C{ (s \log(N/s) )^{7/2}(\log m)^{12}  \over m^{ 1 -  {1 \over 2}\left( {5 \over 6} \right)^{ \lfloor{k/25}\rfloor -1} } }
\] for some absolute constants $ C > 0$ with probability at least $1 - O \left({1 \over m^3} \right)$.

To the best of our knowledge, there is no previous theoretical analysis for NBIHT despite much numerical evidence. There are some results on the convergence of BIHT (unnormalized):
Several works  \cite{boufounos2015quantization,jacques2013quantized,plan2016high} showed that the first iteration of BIHT achieves an approximation error of order of $ {\sqrt{s \log (N/s)} \over m^{1/2} } $. No theoretical analysis for the BIHT convergence has appeared after its first iteration, as far as we are aware, except for \cite{liu2019one}. Theorem V.1 in \cite{liu2019one} states that the iterates of BIHT remain bounded after the first iteration maintaining the same order of accuracy for the rest of iterations. Their analysis, however, does not imply that the approximation error may decrease after the first iteration. In contrast, we establish that the approximation error of NBHIT decays as $ {(s \log(N/s) )^{7/2}(\log m)^{12} \over m}$ after sufficient iterates. On the other hand, Theorem 1 in \cite{jacques2013robust} says that any one-bit compressed sensing recovery algorithms should exhibit error rate of at least order of ${s \over m}$. 
%To our knowledge, there is no previous work showing $O({1 \over m})$ dependence for any polynomial-time one-bit compressed sensing algorithms, so we give the first optimal dependence on $m$ by a BIHT-type algorithm.
Since the lower bound of any reconstruction algorithm error is $\Omega(1/m)$ by that theorem, we give the first optimal dependence on $m$ by a BIHT-type algorithm, or, to our knowledge, any polynomial-time one-bit compressed sensing algorithms.  
Further, the error bound only depends poly-logarithmically on the ambient dimension, $N$, as is expected in the compressed sensing setup.  We believe that the dependence on $s$ and logarithmic factors can be improved and we conjecture that the linear dependence on $s$ of the lower bound is achievable by NBIHT.

%In contrast, we establish that the approximation error of NBHIT decays until it reaches $O\left(  {(s \log(N/s) )^{7/2}(\log m)^{12} \over m} \right)$, achieving the optimal order of decay $\Omega \left({s \over m} \right)$ in \cite{jacques2013robust} when $s$ does not depend on $m$. 


 
%\hfill September 17, 2014

\subsection{Related Works}
There is a vast literature about quantized compressed sensing, e.g., see \cite{powell2013quantization, boufounos2015quantization, chou2016distributed}. As for one-bit compressed sensing,
Plan and Vershynin \cite{plan2013one, plan2012robust} have proposed recovery algorithms based on convex programs. These are computationally efficient and also cover more general classes of signals, but their recovery accuracy is only guaranteed of order ${\sqrt{s \log (N/s) \over m} } $. Knudson et al. \cite{knudson2016one} studied one-bit compressed sensing when the observations are one-bit measurements with dither, i.e., $b = \text{sign}(Ax + \epsilon)$ where $\epsilon$ is a dither noise.  This allows recovery of the manginute of $x$ under certain conditions, but their measurement model is different from ours and the recovery guarantee is still at most of order $ {\sqrt{s \log (N/s) \over m} }$.

%We conclude this section by briefly summarizing relevant papers about quantization other than quantized compressed sensing.  
In the overcomplete settting ($m \ge N$) without assuming sparsity of the signal, \emph{consistent reconstruction} methods produce a signal $w$ whose memoryless scalar quantized measurements $Q(Aw)$ agree with those of $x$. Such methods can be implemented, for example,  using linear programming with constraints enforcing $Q(Aw) = Q(Ax)$ and offers recovery error decay of order ${N \over m} $. Further considerations of computational efficiency have led to iterative signal recovery methods such as the \emph{Rangan-Goyal algorithm} \cite{rangan2001recursive} or the \emph{frame permutation quantization} \cite{nguyen2011frame} in which only one quantized measurement is enforced at each iteration. But their approximation error is still of order  ${N \over m}$ \cite{rangan2001recursive, powell2010mean, nguyen2011frame}, which would not provide a meaningful recovery when $N \gtrsim m$ (the typical compressed sensing setting); Because these algorithms are not specially designed to take into account sparsity, they assume the overcomplete setting  and their error dependency on the signal dimension is $O(N)$ whereas ours is only $O(\log^{7/2}N)$.

\subsection{Notation}
Throughout  this paper, we use the standard notation for Big-O and Big-Omega:
If $f(n) = O(g(n))$,  there exists a positive constant $C>0$ such that 
$|f(n)| \le C|g(n)|$ when $n$ is sufficiently large. On the other hand, if $f(n) = \Omega(g(n))$ means that there exists a positive constant $C>0$ such that 
$|f(n)| \ge C|g(n)|$ for sufficiently large $n$.

We denote the unit sphere and the Euclidean unit ball in $\R^N$ by $\mathbb{S}^{N-1}$ and $B^N_2$ respectively. The Euclidean ball centered at $z$ with radius $r$ is denoted by $B(z,r)$. 
The set of all $s$-sparse vectors in $\R^N$ is denoted by $\Sigma_s^N$ or $K$. We write $T_s$ for the hard thresholding operator onto $K$ that keeps the $s$-largest magnitude elements of a vector and zeros out the rest of its entries. As usual, $\mathcal{N}(0,1)$ denotes the standard Gaussian distribution and $\mathcal{N}(0, I_N)$ the $N$-dimensional standard Gaussian distribution. We assume that the measurement vectors $a_j$ or the rows of $A$ are the standard Gaussian random vectors in $\R^N$. As for norms, $\|\cdot\|$ or $\|\cdot\|_2$ denote the $\ell_2$ norm for a vector and $\|\cdot\|_1$ is the $\ell_1$ norm. All other norms will be defined in later sections. 

Given two sets $A, B \subset \R^N$, $A + B$ is the Minkowski sum of $A$ and $B$, i.e., $A + B = \{v + w  | v \in A, w \in B \}$. The normalized geodesic distance between two vectors $x$ and $y$ is  $d_g(x, y) := {\arccos{\inner{x,y}} \over \pi}$. 

\section{Global convergence of normalized BIHT}
\subsection{The normalized BIHT algorithm}
We briefly describe the normalized binary iterative hard-thresholding (NBIHT) algorithm. The description of the NBIHT algorithm \cite{jacques2013robust} is given in Algorithm \ref{alg:NBIHT}.

\begin{algorithm}
%  \label{alg:NBIHT}
	\caption{Normalized BIHT  }
	\begin{algorithmic}
		\Inputs $m\times N$ matrix $A$, measurements $b=\text{sign}(Ax)$, sparsity level $s$, step size $\tau$
		
		\Initialize	Pick a $x_0$ with $||x_0|| = 1$ and sparsity level $s$.
		
		\Iterate Repeat until a stopping criteria on $k$ is met
		 
		\textbf{Compute $z_{k+1}$}
		\[
		 z_{k+1} := x_k + {\tau \over m} A^T(\text{sign}(Ax) - \text{sign}(Ax_k)) 
		\]
		
		\textbf{Project to $s$-sparse vector sets}\\
		\[
		 \tilde x_{k+1} := T_s(z_{k+1}) 
		\]
		
		\textbf{Normalize}\\
		\[
			x_{k+1} := \tilde x_{k+1}/ ||\tilde x_{k+1}||
		\]
		
		\Outputs The $s$-sparse approximate solution to $\text{sign}(Ax)=b$.

	\end{algorithmic}
	\label{alg:NBIHT}
\end{algorithm}
%Here $T_s$ acts as the projection (hard thresholding) onto the $s$-largest magnitude elements. 

%Note that the algorithm terminates if and only if all the signs of $Ax_k$ agree with those of $Ax$.

An important property of NBIHT is that $z_1$ is an unbiased estimator of $x$ for $\tau = \sqrt{\pi/2}$ (this is true for any initialization $x_0$ that does not depend on $A$). This follows from Proposition \ref{prop:one-bit correlation with Gaussian} which states that for any fixed unit vector $y\in \R^N, \sqrt{\pi \over 2} \mathbb{E} \Big[{1 \over m} A^T \text{sign}(Ay) \Big] = y$. This shows that it is indeed unbiased, i.e., $\mathbb{E}z_1 = x$. Further iterations are no longer unbiased estimators of $x$.  Typically, analyses in this scenario proceed by giving uniform deviations from expectation over the whole set of values that the iterations can take (e.g., RIP).  However, due to the discontinuity of the sign function, uniform bounds on the deviation (which scale according to the distance to $x$) are impossible.  Instead, we give a multiscale approach, showing uniform deviations over a sequence of annuli of increasing radii.  See Section \ref{section:Multiscale analysis} for more details about this idea.

\subsection{Main global convergence theorem}

We are now ready to state our main theorem for the convergence of the NBIHT algorithm. 
\begin{thm}
	\label{thm: Main Convergence Thm}
	Let $x$ be an $s$-sparse unit vector in $\R^N$. Assume that each entry of the measurement matrix $A \in \R^{m \times N}$ is drawn i.i.d. from $\mathcal{N}(0,1)$. Then, there exist positive universal constants $c, C$ such that the NBIHT iterates $x_k$ satisfy
	\[
	\|x_{k} - x\|_2 \le   C{ (s \log(N/s) )^{7/2}(\log m)^{12}  \over m^{ \left(1 -{1 \over 2}\left( {5 \over 6} \right)^{  \lfloor{k/25}\rfloor -1 }  \right) } }
	\] for all $m > \max\{ c \cdot s^{10} \log^{10}(N/s), 24^{48} \}$ with probability at least $1 - O \left({1 \over m^3} \right)$.
\end{thm}

\begin{proof}
	See Theorem \ref{thm:Main_gloabl_convergence3} and its proof.
\end{proof}

By lettting the iteration number $k \rightarrow \infty$, Theorem \ref{thm: Main Convergence Thm} imples the following Corollary about the approximation error of NBIHT. 

\begin{cor}
	\label{cor:Main_global_convergence_Cor}
	Under the same assumptions in Theorem \ref{thm: Main Convergence Thm}, the NBIHT iterates $x_k$ converge to $x$ up to approximation error $  C{ (s \log(N/s) )^{7/2}(\log m)^{12}  \over m} $ with probability at least $1 - O \left({1 \over m^3} \right)$ provided $m > \max\{ c \cdot s^{10} \log^{10}(N/s), 24^{48} \}$.
\end{cor}

\begin{rem}
	Jacques et.al.  \cite{jacques2013robust} lower bounded the best achievable reconstruction error from one-bit (sign) measurements. Specifically, Theorem 1 in \cite{jacques2013robust} shows that the approximation error decay rate cannot be faster than $\Omega({s \over m})$ no matter what algorithm we use as long as $m = \Omega(s^{3/2})$. Thus, Corollary \ref{cor:Main_global_convergence_Cor} reveals that the NBIHT algorithm achieves the best possible approximation error decay rate in $m$. 
\end{rem}

\begin{rem}
    The number of measurements to have the optimal error decay rate in our theory is substantially large. We have focused on proving the optimal error decay rate not optimizing the constant. 
\end{rem}

%\section{Related work}

\section{Multiscale analysis based on restricted approximate invertibility condition}
\label{section:Multiscale analysis}

This section illustrates the idea behind the proof of our main convergence result. First, we introduce the dual norm which is used in the rest of the paper.   

\begin{dfn}
	Let $G$ be a symmetric subset in $\R^N$. The dual norm of $G$ is given by
	\[
	\|x\|_{G^\circ} := \sup\limits_{u \in G} \inner{x,u}, \quad x \in \R^N.
	\] We also denote $(G-G) \cap dB^N_2$ by $G_d$. Moreover, when $G$ is a cone, we see that  $\|x\|_{G_t^\circ} = t 	\|x\|_{G^\circ}$ for any $t > 0$. 
\end{dfn} 

Recall that in the standard compressed sensing, one wishes to reconstruct an $s$-sparse vector $x$ from its linear measurements $Ax$. There are numerous algorithms available for this problem including \emph{iterative hard thresholding (IHT)}, which is commonly used due to its simplicity and computational efficiency. We present IHT in Algorithm \ref{alg:IHT} as it has motivated  BIHT-type algorithms and our NBIHT convergence proof hinges, in part, on the analysis of IHT.

\begin{algorithm}
	\caption{IHT }
	\begin{algorithmic}
		\Inputs $m\times N$ standard Gaussian random matrix $A$, measurements $b=Ax$, sparsity level $s$
		
		\Initialize	Pick a $x_0$ with sparsity level $s$.
		
        \Iterate Repeat until a stopping criteria on $k$ is met
		 
		\textbf{Compute $z_{k+1}$}
		\[
		 z_{k+1} := x_k + {1 \over m} A^T(Ax - Ax_k) 
		\]
		
		\textbf{Project to $s$-sparse vector sets}\\
		\[
		 x_{k+1} := T_s(z_{k+1}) 
		\]
		
		\Outputs The approximate $s$-sparse solution to $Ax=b$.
		
	\end{algorithmic}
	\label{alg:IHT}
\end{algorithm}

The proofs for the linear convergence of IHT typically start with showing $\| z_{k+1} -x \|_2 \le c \| x_k -x  \|_2$ for some $0 < c < 1$. Since we want to convey the idea, let us assume $c = {1 \over 8}$, i.e., we have $\| z_{k+1} -x \|_2 \le {1 \over 8} \| x_k -x  \|_2$. Then, a standard argument implies a contraction inequality $\|x_{k+1} - x \|_2 \le {1 \over 2} \| x_k -x\|_2$ which leads to the linear convergence of IHT. To prove the inequality $\| z_{k+1} -x \|_2 \le {1 \over 8} \| x_k -x  \|_2$, define ${\tilde A} := {1 \over \sqrt{m}}A$, which satisfies the RIP with high probability for sufficiently large $m$ \cite{foucart2013invitation}.
Since $z_{k+1} -x = {1 \over m} A^T(Ax - Ax_k) - (x - x_k)  =  {\tilde A}^T({\tilde A} x - {\tilde A} x_k) - (x - x_k) $, it is evident that controlling the term ${\tilde A}^T({\tilde A} x - {\tilde A} x_k) - (x - x_k) $ plays a critical role in the convergence analysis. This term is controlled with the RIP which implies that ${\tilde A}^T{\tilde A}$ acts as approximately the identity when restricted to sparse vectors.  The RIP is equivalent to the following equation \cite{foucart2013invitation} which will naturally connect to the RAIC:
\[
\sup_{y \in K} 	 \left\| {1 \over m} \cdot { A}^T({ A} x - { A} y) - (x - y)  \right \|_{K_1^\circ} = \sup_{y \in K} 	 \left\| {\tilde A}^T({\tilde A} x - {\tilde A} y) - (x - y)  \right \|_{K_1^\circ}  \le {1 \over 8}||x-y||_2.
\] This is the key inequality in the proof for the linear convergence proof of IHT. 

Unfortunately, due to the discontinuity of  $\text{sign}(A\cdot)$, the measurement operator in one-bit compressed sensing does not have such handy properties; Arbitrarily close two vectors can be mapped to points which are ``far" under the operator $\text{sign}(A\cdot)$. However, it turns out that the randomness of $\text{sign}(A \cdot)$ still allows us to derive a ``restricted approximate invertibility condition (RAIC)" which makes it possible to replace the RIP to facilitate the analysis and we think it is interesting on its own. Informally speaking, the RAIC is similar to RIP but holds only on an annulus at $x$ ($r_{i+1} \le ||x-y||_2 \le r_i$) with a small additive error. To be more precise, the RAIC with the associated parameters $r_i$ and $r_{i+1}$ holds if
\[
\sup_{\substack{y \in K \cap \mathbb{S}^{N-1} \\ r_{i+1} \le ||x-y||_2 \le r_i }} 	{\left\| {\tau \over m} \cdot A^T(\text{sign}(Ax) - \text{sign}(Ay)) -  (x-y) \right \|_{K_1^\circ}   }  \le {1 \over 8}||x-y||_2 + {3 \over 50} r_{i+1}.
\] 
Here $\{r_j\}_{j=1}^L$ is a finite decreasing  sequence of real numbers with $r_1 \approx  {\sqrt{s \log (N/s)} + \sqrt{\log m}  \over m^{1/2}} $ and $r_L \lesssim {(s \log(N/s) )^{7/2}(\log m)^{12} \over m}$, whose exact specification will be given later. It is well-known that the first iteration of the NBIHT, say $x_1$, satisfies $\|x_1 - x\| \le r_1$ with high probability. 

Then, the rest of analysis is based on a contraction-type inequality for the approximation error (up to a small additive term) that is obtained by applying the RAIC with the current error scale; For example, if the $k$-th NBIHT iterate, $x_k$ satisfies $r_{i+1} \le ||x-x_k||_2 \le r_i$ for some $1 \le i \le L$, then we apply the RAIC with parameters $r_i$ and $r_{i+1}$. Repeated applications of the RAIC combined with an elementary argument imply that $x_k$ converges to $x$ until it reaches the approximation error scale $r_L$. Lastly, once $x_k$ reaches the error scale $r_L$, we show that the rest of the NIBHT iterates stay close to $x$ with the same error scale as $r_L$.

\section{Outline of the proof of main technical results}

In this section, we present the proof outline of the main technical results including the RAIC.

\subsection{Sketch of the proof of our main convergence result based on the RAIC}
\begin{enumerate} 
	
	\item Let $K$ be the set of $s$-sparse vectors. First, from the gradient step iteration for the NBIHT, $z_{k+1} = x_k + {\tau \over m} A^T(\text{sign}(Ax) - \text{sign}(Ax_k))$. Using a standard argument for the projections to $K\cap \mathbb{S}^{N-1}$, we will see that
	\begin{align}
	\label{outline_step1}
	\|x_{k+1} -x \|_2 \le 4 \left\|  {\tau \over m}\cdot A^T(\text{sign}(Ax) - \text{sign}(Ax_k)) - (x-x_k) \right \|_{K_1^\circ}.
	\end{align}

	\item 	Let $L$ be the largest positive integer satisfying $m^{{1 \over 40} ({5 \over 6})^L} > 24$. Such $L$ always exists from the assumption $m > 24^{48} = 24^{40\left({6 \over 5}\right)}$. Suppose $i$ is a postive integer less than $L$. We define a sequence of sets, $K^{(i)} := \{y \; |\; \text{y is an $s$-sparse unit vector, } r_{i+1} \le \|y-x\| \le r_i\}$ where $\{r_i\}$ is a decreasing sequence of real numbers satisfying
	\[
	r_i \lesssim m^{- \left(1 -{1 \over 2}\left({5 \over 6} \right)^{i-1} \right)}  (\log m)^{10} (s \log (N/s))^3.
	\] By the choice of $L$ and the sequence $\{r_i\}$, it turns out that $r_L \lesssim {(s \log(N/s) )^{7/2}(\log m)^{12} \over m}$. The exact specifications of $r_i$ will be given in Section \ref{section:Toward to the proof of NBIHT convergence}.
	
	\item 
	The RAIC says that for all $y \in K^{(i)}$, 
	\begin{align}
	\label{ineq:RAIC_outline}
	\left\|  {\tau \over m}\cdot A^T(\text{sign}(Ax) - \text{sign}(Ay)) - (x-y) \right \|_{K_1^\circ} \le {3 \over m^{{1 \over 40} ({5 \over 6})^{i}}} \|x - y\|_2 + {3 \over 50}r_{i+1}
	\end{align}
	 with high probability. This bound and \eqref{outline_step1} imply that
	\[
	\|x_{k+1} -x \|_2  \le {12 \over m^{{1 \over 40} ({5 \over 6})^{i}}} \|x - x_k\|_2 +  {6 \over 25}r_{i+1}
	\] for all $x_k$ with $r_{i+1} \le \|x-x_k\| \le r_i$. Since $m^{{1 \over 40} ({5 \over 6})^L} > 24$ and $i < L$, $m^{{1 \over 40} ({5 \over 6})^i} > 24$. Thus,  
	\[
	\|x_{k+1} -x \|_2  \le {1 \over 2} \| x_k - x\|_2 +  {6 \over 25} r_{i+1}.
	\]
	Note that ${1 \over 2} \|x_k - x\|_2 + {6 \over 25} r_{i+1} \le {1 \over 2} r_i + {1 \over 2} r_i \le r_i$, so $\|x_{k+1} -x\|_2 \le r_i$. Then, a simple induction argument yields the following: For any $t \ge 1$, either there exists $p$ with $1 \le p \le t$ such that $\|x_{k+p} -x\| \le r_{i+1}$ or 
	\[
	\|x_{k+t} -x \|_2  \le \left({1 \over 2}\right)^t \|x - x_k\|_2 +  {12 \over 25} r_{i+1}
	\] with high probability.
	
	From the relation between $r_i$,$r_{i+1}$ and the previous inequality, one can easily show that $\|x_{k+p} -x\| \le r_{i+1}$ for some $p$ with $1 \le p \le 25$. In other words, $25$ iterations are enough for the BIHT iterates to reach the next level $K^{(i+1)}$. Repeating this argument yields that $x_k$ approaches $x$ until it reaches the level set $K^{(L)}$, guaranteeing the approximation error $\approx r_L$ up to some logarithmic factors. 
	This is the main idea of Theorem \ref{thm: Main Convergence Thm}.
	
\end{enumerate}
	
\subsection{Sketch of the proof of the RAIC}	
It remains to present the idea of the proof the RAIC, which turns out to be the main challenge in our approach. Here are the key ingredients of the proof.
\begin{itemize}
	
	\item Local Binary Embedding (LBE) \cite{oymak2015near}: The LBE is a refined version of the Binary $\epsilon$-Stable Embedding (B$\epsilon$SE) by Jacques et.al. \cite{jacques2013robust}. The B$\epsilon$SE states that the distance of any two $s$-sparse unit vectors is close to the normalized Hamming distance of their one-bit measurements up to additive error $\epsilon$. It is also proved that the B$\epsilon$SE holds for an $m \times N$ Gaussian random measurement matrix with $\epsilon$ of order of $\sqrt{s/M}$ with high probability, but we were not able to apply B$\epsilon$SE to achieve our decay rate essentially because $\sqrt{1 \over M} \gg {1 \over M}$.
	
	On the other hand, although the LBE resembles the B$\epsilon$SE, the key difference is that its additive error in the embedding becomes smaller as two vectors are sufficiently close to each other. In other words, we have more accurate control of the additive error if one vector belongs to a small neighborhood of the other.  
	
	\item The standard $\epsilon$-net argument is used for a uniform approximation of Lipschitz continuous parts in our analysis, whereas the LBE is applied to the discontinuous part involving the $\text{sign}(\cdot)$ operation.
	
	\item We also use a concentration argument based on the bounded Bernstein's inequality in several places. 
\end{itemize}

\begin{enumerate} 
	\item The first step of our proof is the orthogonal decomposition of a Gaussian random vector into three components along the directions of $x-x_k$, $x+x_k$ and their orthogonal complement to facilitate the analysis. More precisely, given a unit vector $y$, consider the following decomposition of each measurement vectors $a_i$
	\[
	a_i = \inner{a_i, {x -y \over \|x -y \| }} {x -y \over \|x -y \| } + \inner{a_i,{x +y \over \|x +y \|}  }{x +y \over \|x +y \| } + b_i(x,y),
	\] where $b_i(x,y)$ is the component of $a_i$ orthogonal to $x -y$ and $x +y$. This decomposition allows us to decouple the right hand side of \eqref{ineq:RAIC_outline} into three parts, each of which has a similar structure. 
	As will be made more precise in later sections, the analysis of each part essentially boils down to controlling the following form of a sum:
	% with accuracy at least $\epsilon = O \left( {\sqrt{s \log (2N/s) } \over m^{2/3} } \right)$:
	\begin{align}
		\label{eq:Multiplier_process_sum}
		{\tau \over m} \sum\limits_{i=1}^m \left [( \text{sign}(\inner{a_i,x} ) - \text{sign}(\inner{a_i,y})) g(a_i, y) \right ]	- \mathbb{E}_{a} \left[ ( \text{sign}(\inner{a,x} ) - \text{sign}(\inner{a,y})) g(a, y)\right]
	\end{align}
	for all $s$-sparse unit vectors $y$. Here $g$ is a jointly $1$-Lipschitz continuous function bounded by $C(\sqrt{s \log (2N/s)} + \sqrt{\log m})$ for some fixed constant $C \ge 1$. For the sake of illustration, we further assume $\mathbb{E}_{a} \left[ ( \text{sign}(\inner{a,x} ) - \text{sign}(\inner{a,y})) g(a, y)\right] = 0$ for all unit vector $y$.

    \item  We are now ready to present the idea on how to control the sum  \eqref{eq:Multiplier_process_sum} by giving strings of inequalities with technical details followed by justifications. First, let $\mathcal{N}_{K^{(i)}}$ be an $\epsilon$-net of $K^{(i)}$ with $\epsilon := \delta_i$ which will be defined in Section \ref{section:Toward to the proof of NBIHT convergence}. So, for any $y \in K^{(i)}$, there exists $\hat{y} \in \mathcal{N}_{K^{(i)}}$ with $\|y-\hat{y}\|_2 \le \delta_i$.
    Then, 
	\begin{align*}
			& \left | {\tau  \over m} \sum\limits_{i=1}^m  [( \text{sign}(\inner{a_i,x} ) - \text{sign}(\inner{a_i,y})) g(a_i, y) \right | \\
			&  \stackrel{(i)}{\le} \left| {\tau  \over m} \sum\limits_{i=1}^m \left [( \text{sign}(\inner{a_i,x} ) - \text{sign}(\inner{a_i,y}) g(a_i, y) \right ] - {\tau  \over m} \sum\limits_{i=1}^m \left [( \text{sign}(\inner{a_i, x } ) - \text{sign}(\inner{a_i,\hat{y}})) g(a_i, y) \right ] \right | \\
			& \qquad +\left | {\tau  \over m} \sum\limits_{i=1}^m  [( \text{sign}(\inner{a_i,x} ) - \text{sign}(\inner{a_i,\hat{y}})) g(a_i, y) \right | \\
			&  \stackrel{(ii)}{\le}   {\tau  \over m} \sum\limits_{i=1}^m \left | \text{sign}(\inner{a_i, y } ) - \text{sign}(\inner{a_i,\hat{y}}) \right| |g(a_i, y)| +\left | {\tau  \over m} \sum\limits_{i=1}^m  ( \text{sign}(\inner{a_i,x} ) - \text{sign}(\inner{a_i,\hat{y}})) g(a_i, y) \right | \\
			 & \stackrel{(iii)}{\le}   2 C\tau  \cdot (\sqrt{s \log (2N/s)} + \sqrt{\log m}) \log m \cdot \delta_i + \left | {\tau  \over m} \sum\limits_{i=1}^m  ( \text{sign}(\inner{a_i,x} ) - \text{sign}(\inner{a_i,\hat{y}})) g(a_i, y) \right | \\
			& \stackrel{(iv)}{\le}   2 C \tau \cdot (\sqrt{s \log (2N/s)} + \sqrt{\log m}) \log m \cdot \delta_i \\
			&\qquad + {\tau  \over m} \sum\limits_{i=1}^m  \left | \text{sign}(\inner{a_i,x} ) - \text{sign}(\inner{a_i,\hat{y}})\right | |g(a_i, y) -  g(a_i, \hat{y}) | \\
			&\qquad +   \left | {\tau  \over m} \sum\limits_{i=1}^m  ( \text{sign}(\inner{a_i,x} ) - \text{sign}(\inner{a_i,\hat{y}})) g(a_i, \hat{y}) \right | \\
			&\stackrel{(v)}{\le}  4 C \tau \cdot (\sqrt{s \log (2N/s)} + \sqrt{\log m}) \log m \cdot \delta_i  +   \left | {\tau  \over m} \sum\limits_{i=1}^m  ( \text{sign}(\inner{a_i,x} ) - \text{sign}(\inner{a_i,\hat{y}})) g(a_i, \hat{y}) \right |\\	
			&\stackrel{(vi)}{\le}  4 C \tau \cdot (\sqrt{s \log (2N/s)} + \sqrt{\log m}) \log m \cdot \delta_i  + {2 \over m^{{1 \over 40} ({5 \over 6})^{i}}}\|x - \hat{y}\|_2 + {1 \over 600} r_{i+1}\\
			&\stackrel{(vii)}{\le}   {2 \over m^{{1 \over 40} ({5 \over 6})^{i}}}\|x - \hat{y}\|_2 + {4 \tau \over 600} r_{i+1} \\
			&\stackrel{(viii)}\le  {2 \over m^{{1 \over 40} ({5 \over 6})^{i}}}\|x - y\|_2 + 2 \|y - \hat{y} \| + {4 \tau \over 600} r_{i+1} \\
			&\stackrel{(ix)}\le  {2 \over m^{{1 \over 40} ({5 \over 6})^{i}}}\|x - y\|_2 + {3 \over 200} r_{i+1}.
	\end{align*}
	Here are justifications for the above chain of inequalities: (i), (ii), (iv), and (viii) follow from the triangle inequality. (iii) is due to the local binary embedding, the $\epsilon$-net $\mathcal{N}_{K^{(i)}}$, and the fact that $g$ is bounded by $C(\sqrt{s \log(2N/s)} + \sqrt{\log m})$. For (v), we used a standard $\epsilon$-net argument since $g$ is $1$-Lipschitz. Next, (vi) follows from the Bernstein's inequality for mean-zero bounded random variables. This step is important since we exploit the cancellation of terms in the sum based on a concentration inequality. Simple arguments without considering the cancellation effect in the sum would yield  $C(\sqrt{s \log (2N/s)} + \sqrt{\log m}) \log m \cdot r_{i+1} > r_{i+1}$, which is not good enough to show the RAIC used in \eqref{ineq:RAIC_outline}. Lastly, (vii) and (ix) are from the relation between $\delta_i$ and $r_{i+1}$. 
\end{enumerate}

\section{Toward the proof of NBIHT convergence}
\label{section:Toward to the proof of NBIHT convergence}
Several previous works in one-bit compressed sensing are based on the binary stable embedding property (BSE) type of inequalities \cite{jacques2013robust,jacques2013quantized,plan2012robust}. Although this property is interesting on its own, it is not as strong as the RIP in the standard compressed sensing when we have full linear measurements for sparse signals, so we were not able to use the BSE to achieve our goal. Instead, our proof for the main theorem relies on various different tools and this section provides them before the proof appears. 

We begin with the following proposition relating the approximation error of NBHIT iterates $x_k$ to $T_s(z_k)$ in Algorithm \ref{alg:NBIHT}.
\begin{prop}
    \label{prop:First_step_recurrence}
	If $T_s$ is the hard thresholding operator and $x_k = T_s(z_k)/ \|T_s(z_k)\|$, then
	\begin{align}
	\label{eq:First_step_recurrence}
	\|x_k- x\|_2 &\le \|T_s(z_k) - x_k\|_2 + \|T_s(z_k) - x\|_2  \le 2\|T_s(z_k) - x\|_2.
	\end{align} 
\end{prop}
\begin{proof}
	The first inequality is by the triangle inequality. The second inequality follows from the facts that $\|x\| = \|x_k\| = 1$ and $x_k$ is the closest point on the unit sphere to $T_s(z_k)$.
\end{proof}

\subsection{The first iteration of normalized BIHT}
\label{section:FirstIteration} 
The first iteration of the BIHT is investigated by several papers \cite{jacques2013quantized, plan2016high}. In particular, Proposition 1 and 2 in \cite{jacques2013quantized} state that $\|T_s(z_1) - x\| = O \left( {\sqrt{s \log (N/s)} + \sqrt{\log m}  \over m^{1/2}}\right)$ with high probability. The following statement makes this precise.

\begin{prop}
    \label{prop:First_iteration_bound}
	Let $A$ be a $m \times N$ Gaussian random matrix. Then, there exists constant $\bar{c} \ge 1$ such that 
	\[
	\|T_s(z_1) -x \| \le {1 \over 2} \cdot {\sqrt{s \log (N/s)} + \sqrt{\log m} \over m^{1/2}}
	\] with probability at least $1 -  \bar{c} \left( ({s \over N})^{5 s} \cdot {1 \over m^5} \right)$.

\end{prop}
Propositions \ref{prop:First_step_recurrence} and \ref{prop:First_iteration_bound} together imply that
\begin{align}
\label{eq:bound_for_first_NBHIT_iteration}
\| x_1 - x \| \le  {\sqrt{s \log (N/s)} + \sqrt{\log m} \over m^{1/2}}
\end{align} with probability at least $ 1 - \bar{c} \left( ({s \over N})^{5 s} \cdot {1 \over m^5}\right)$.
Also note that $x_1$ is a unit $s$-sparse vector since it is an iterate of NBIHT. 

\subsection{Metric projection}
In this section, we introduce the notion of the restricted approximate invertibility condition (RAIC) and its connection to the approximation error $\|x_k - x\|$. We start with the definition of Gaussian width of a set.

\begin{dfn}
	The Gaussian width of a set $G \in \R^N$ is defined by
	\[
		w(G) := \mathbb{E} \sup\limits_{u \in G} \inner{h,u} 
	\] where $h \sim \mathcal{N}(0, I_N)$. Note that $w(G) = \mathbb{E} \|h\|_{G^\circ}$.
\end{dfn} 

The following Corollary 8.3. in \cite{plan2016high} allows us to control $\|x_{k+1} -x \|_2$.
\begin{cor}
\label{cor: Metrical_projection bound}
	Let $P_G$ be the orthogonal projection on a star-shaped set $G$. Then, for $z \in G$, we have
	\[
		\|P_G(w) -z \|_2 \le \max \left(t, {2 \over t} \|w-z\|_{G_t^\circ} \right) \quad \text{ for any $t > 0$}.
	\]
\end{cor}

Since $T_s$ is the hard thresholding operator to the $s$-sparse vectors, by Corollary \ref{cor: Metrical_projection bound}, we have

\begin{align*}
\left\| T_s(z_{k+1}) - x \right \|_2 
& \le \max \left(t, {2 \over t} \left\| z_{k+1}- x \right \|_{K_t^\circ} \right) \\
& \le \max \left(t, {2 \over t} \cdot t \left\| z_{k+1}- x \right \|_{K_1^\circ} \right) \\
& \le 2 \left\| z_{k+1}- x \right \|_{K_1^\circ},
\end{align*} where the second inequality follows from the fact that $K$ is a symmetric cone and the third one is by taking $t = 2 \left\| z_{k+1}- x \right \|_{K_1^\circ}$. 

Hence, combining the above inequality, Proposition \ref{prop:First_step_recurrence}, and $z_{k+1} = x_k+ {\tau \over m}\cdot A^T(\text{sign}(Ax) - \text{sign}(Ax_k))$ yields
\begin{align}
	\label{ineq:Main_inequality}
	\|x_{k+1} -x \|_2 \le 4 \left\|  {\tau \over m}\cdot A^T(\text{sign}(Ax) - \text{sign}(Ax_k)) - (x-x_k) \right \|_{K_1^\circ}.
\end{align}

\begin{rem}
	Note that if we had linear measurements of $x$, not the one-bit measurements, the restricted isometry property (RIP) would be sufficient to establish a contraction, that is, $\|x_{k+1} -x \|_2 \le \rho \|x_k - x\|_2$ for some $0< \rho <1$. Indeed, 
	the $ \delta_{2s}$-RIP for matrix ${1 \over \sqrt{m}}A$ in can be recasted as 
	\[
	\max_{S \subset [N], \text{card}(S) \le 2s } \left \| {1 \over m} A^T|_S A|_S - I \right \|_{2 \rightarrow 2} = \delta_{2s}
	\] which implies that
	\[
	\left\| {1 \over m} A^T(Ax - Ay) -  (x-y) \right \|_{K_1^\circ}  \le \delta_{2s} ||x-y||_2
	\] for all $s$-sparse vectors $x$ and $y$ (See \cite{foucart2013invitation} for more details).
	
	Hence, if we had  $	\left\| {1 \over m} A^T(Ax - Ax_k) -  (x-x_k) \right \|_{K_1^\circ}$ instead of $\left\|  {\tau \over m}\cdot A^T(\text{sign}(Ax) - \text{sign}(Ax_k)) - (x-x_k) \right \|_{K_1^\circ}$ in the right hand side of \eqref{ineq:Main_inequality}, the RIP would give a contraction for sufficiently small $\delta_{2s}$, which leads to a linear convergence. 
\end{rem}

Unfortunately, because of a severe discontinuity of the operator $\text{sign}(A \cdot)$, we don't have the RIP. However we will show that the following property still holds, which provides ``local approximate version of RIP". 

\begin{dfn}[Restricted Approximate Invertibility Condition] \mbox{}\\
	We say that the matrix $M$ satisfies the $(\nu, \delta_{2s}, \eta_{2s}, r_{lb}, r_{ub})$-restricted approximate invertibility condition (RAIC) at $x$ (the ground truth $s$-sparse unit vector) if the following inequality holds for all $s$-sparse unit vectors $y$ with $r_{lb} \le \|x-y\|_2 \le r_{ub}$,
	\[
	\left\| \nu \cdot M^T(\text{sign}(Mx) - \text{sign}(My)) -  (x-y) \right \|_{K_1^\circ} \le \delta_{2s} ||x-y||_2 + \eta_{2s}.
	\]
\end{dfn}

Note that the RAIC is similar to the RIP except it holds for a certain region at $x$ with an additive error $\eta_{2s}$. As will be more clear in Proposition \ref{prop:Main_proposition}, we have the RAIC for the matrix $A$ with high probability within certain regions around $x$. At this point, one would notice that the RAIC can be applied to \eqref{ineq:Main_inequality} under appropriate conditions.

\subsection{Orthogonal decomposition of measurement vectors}

We will use the following orthogonal decomposition of Gaussian measurement vectors $a_i$, which is inspired by the decomposition technique in Plan et. al \cite{plan2016high}, as the first step in the proof of Theorem \ref{thm: Main Convergence Thm} in the next section.

\begin{lem}
	\label{lem:Orthogonal_decomposition_Gaussian_vector}
	Suppose that $a_i$'s are the standard Gaussian random vectors. 
	Let $x$, $y$ be unit vectors in $\R^N$. Since $x -y$ and $x +y$ are orthogonal to each other, then we have
	\[
	a_i = \inner{a_i, {x -y \over \|x -y \| }} {x -y \over \|x -y \| } + \inner{a_i,{x +y \over \|x +y \|}  }{x +y \over \|x +y \| } + b_i(x,y)
	\] where $b_i(x,y)$ is the component of $a_i$ orthogonal to $x -y$ and $x +y$.
\end{lem}

\begin{proof}
	Since $x, y$ are unit vectors, one can easily check that $\inner{a_i, {x -y \over \|x -y \| }} {x -y \over \|x -y \| } $ and $ \inner{a_i,{x +y \over \|x +y \|}  }{x +y \over \|x +y \| }$ are orthogonal by a direct calculation. 
\end{proof}

\subsection{Uniform bounds}
\subsubsection{Uniform bounds for \texorpdfstring{ $ \left| \inner{a_i,{x +y \over \|x +y \|}  } \right|, \left| \inner{a_i,{x -y \over \|x -y \|}  } \right|, \|b_i(x,y)\|_{K_1^\circ}$} {}}
Since we will use the Bernstein's inequality for the bounded random variables later, we need to show $ \left| \inner{a_i,{x +y \over \|x +y \|}  } \right|$, $ \left| \inner{a_i,{x -y \over \|x -y \|}  } \right|$, $\|b_i(x,y)\|_{K_1^\circ}$ bounded with high probability.

First, the following upper bound of $\|a\|_{K_1^\circ}$ for the standard Gaussian random vector $a$ will be used repeatedly in this subsection. 
\begin{prop}
	\label{prop:Uniform_bound_for_dual_norm_Gaussian_vector}
	Suppose $a \sim \mathcal{N}(0,I_N)$. Then there exist absolute constants $C_b \ge 1$ and $0 < c_b \le 1$ such that with probability at least $1 - {2 \over m^5}$, we have
	\[
	\|a\|_{K_1^\circ} \le C_b  \sqrt{ s\log(N/s) } + \sqrt{ {5 \log m \over c_b}}.
	\]
\end{prop}

\begin{proof}
	From the definition of the dual norm $\|a\|_{K_1^\circ}$, 
	\[
	\mathop{\mathbb{E}}_{a} \|a\|_{K_1^\circ}  = \mathop{\mathbb{E}}_{a} \sup\limits_{\substack{v -w \in K-K \\ \|v - w\|_2 \le 1 }}  \inner{a, v-w} = w(K_1) \le C_b  \sqrt{ s\log(N/s) }
	\] where the inequality is by a well-known bound of Gaussian width of the set of $2s$-sparse unit vectors for some constant $C_b \ge 1$. 
	
	Since $a \rightarrow \|a\|_{K_1^\circ}$ is a Lipschitz continuous function with Lipschitz constant at most $1$,  by the Gaussian concentration inequality \cite{ledoux2001concentration, vershynin2018high}, 
	\begin{align}
	\label{eq: Uniform_bound 0}
	\left \| a \right \|_{K_1^\circ}  \le C_b \sqrt{ s\log(N/s) } + u
	\end{align}
	with probability at least $1 - 2\exp(-c_b u^2)$ for some constant $0 < c_b \le 1$. Finally, set $u = \sqrt{ {5 \log m \over c_b}}$ to obtain the proposition.
\end{proof}

Let $E_i$ be the event that the bound in Proposition \ref{prop:Uniform_bound_for_dual_norm_Gaussian_vector} holds for $a_i$. Note that $\mathbb{P}(E_i) \ge 1 - {2 \over m^5}$
and also define $E_u$ such that
\[
E_u = \cap_{i=1}^m E_i.
\] Then, by the union bound, $\mathbb{P}(E_u) \ge 1 - {2 \over m^4}$.

\begin{lem}
	\label{lem:Uniform_bound 1}
	With probability at least $1 - {2 \over m^4}$, we have
	\begin{align*}
	    & \max_{1\le i \le m} \sup\limits_{y \in K \cap \mathbb{S}^{N-1}} \left \{ \left| \inner{a_i,{x +y \over \|x +y \|}  } \right|,  \left| \inner{a_i,{x -y \over \|x -y \|}  } \right|, \|b_i(x,y)\|_{K_1^\circ} \right \} \\
	&\quad \le   3\left(C_b  \sqrt{ s\log(N/s) } + \sqrt{ {5 \log m \over c_b}}\right).
	\end{align*}

\end{lem}

\begin{proof}
	For any $y \in K \cap \mathbb{S}^{N-1}$, we observe that ${x +y \over \|x +y \|} \in (K - K) \cap \mathbb{S}^{N-1}$ because $K$ is symmetric. 
	%From the definition of the Gaussian width, 
	\begin{align*}
	& \sup\limits_{y \in K \cap \mathbb{S}^{N-1}}  \left| \inner{a_i,{x +y \over \|x +y \|}  } \right| \\
	& \quad \le \sup\limits_{w,y \in K \cap \mathbb{S}^{N-1}}  \left| \inner{a_i,{w +y \over \|w +y \|}  } \right| \\
	& \quad \le   \sup\limits_{u -v \in (K-K), \|u-v\|_2 \le 1 }  \left| \inner{a_i, u -v}  \right| \\
	& \quad = \sup\limits_{u -v \in (K-K), \|u-v\|_2 \le 1 }  \inner{a_i, u -v} \\
	& \quad = \|a_i \|_{K_1^\circ} \le C_b  \sqrt{ s\log(N/s) } + \sqrt{ {5 \log m \over c_b}}
	\end{align*}  under the event $E_u$. Here the second equality holds since $K$ is symmetric and the last inequality is from Proposition \ref{prop:Uniform_bound_for_dual_norm_Gaussian_vector}. 
	
	Again by the symmetry of $K$, this also implies that
	\begin{align*}
	& \sup\limits_{y \in K \cap \mathbb{S}^{N-1}}  \left| \inner{a_i,{x -y \over \|x -y \|}  } \right| \le C_b  \sqrt{ s\log(N/s) } + \sqrt{ {5 \log m \over c_b}}.
	\end{align*}

	As for the bound for $\left \| b_i(x, y) \right \|_{K_1^\circ}$, we first start from the decomposition of $a_i$ in Lemma \ref{lem:Orthogonal_decomposition_Gaussian_vector}.
	\[
	 b_i(x,y) = \inner{a_i, {x -y \over \|x -y \| }} {x -y \over \|x -y \| } + \inner{a_i,{x +y \over \|x +y \|}  }{x +y \over \|x +y \| } - a_i. 
	\] 
	After taking the dual norm $\| \cdot \|_{K_1^\circ}$ on both sides, we have
	\begin{align*}
		\| b_i(x,y) \|_{K_1^\circ}
		& \le \left\|  \inner{a_i, {x -y \over \|x -y \| }} {x -y \over \|x -y \| } \right\|_{K_1^\circ}+ \left\| \inner{a_i,{x +y \over \|x +y \|}  }{x +y \over \|x +y \| } \right\|_{K_1^\circ}+  \|a_i \|_{K_1^\circ} \\
		& \le  \left| \inner{a_i, {x -y \over \|x -y \| }} \right| \cdot  \left\|  {x -y \over \|x -y \| } \right\|_{K_1^\circ} + \left| \inner{a_i,{x +y \over \|x +y \|}  } \right| \cdot \left\| {x +y \over \|x +y \| } \right\|_{K_1^\circ} +  \|a_i \|_{K_1^\circ} \\
		& =  \left| \inner{a_i, {x -y \over \|x -y \| }} \right|  + \left| \inner{a_i,{x +y \over \|x +y \|}  } \right| +  \|a_i \|_{K_1^\circ} \\
		& \le 3\left(C_b  \sqrt{ s\log(N/s) } + \sqrt{ {5 \log m \over c_b}}\right),
	\end{align*} where the first and second inequalities are by the triangle inequality, the equality is from the definition of the dual norm $\| \cdot \|_{K_1^\circ}$, and we applied the bounds for  $\left| \inner{a_i, {x -y \over \|x -y \| }} \right|, \left| \inner{a_i,{x +y \over \|x +y \|}  } \right|$, and  $\|a_i \|_{K_1^\circ}$ in the last inequality. Thus we have the bound in the lemma under the event $E_u$ which holds with probability exceeding $1 - {2 \over m^4}$.

\end{proof}

\subsubsection{Other uniform bounds}

\begin{prop}
	\label{prop:difference_of_two_normalized_vectors}
	Let $y, \hat{y} \in \mathbb{S}^{N-1}$ with $r \le \|x-y\| $ and $r \le \|x- \hat{y}\|$ for some $r > 0$. Also, assume that $\|y-\hat{y}\| \le \delta$. Then, we have
	\[
	\left \|{x -y \over \|x -y \| } -  {x -\hat{y} \over \|x -\hat{y}  \| } \right \| \le {2 \delta \over r}.
	\]
\end{prop}
\begin{proof}
	Repeated applications of triangle inequality yield the following chain of inequalities. 
	\begin{align*}
	\left \|{x -y \over \|x -y \| } -  {x-\hat{y} \over \|x -\hat{y}  \| } \right \| 
	&\le \left \| 	{x -y \over \|x -y \| }  - 	{x -y \over \|x - \hat{y} \| } \right \| + \left \| 		{x -y \over \|x - \hat{y} \| } -  {x -\hat{y} \over \|x -\hat{y}  \| } \right \|\\
	&\le \left | 	{1\over \|x -y \| }  - 	{1 \over \|x- \hat{y} \| } \right | \|x-y\| + {1 \over \|x -\hat{y}  \| } \left \| x-y -  (x -\hat{y} )  \right \|\\
	&\le \left | 	{\|x - \hat{y} \| - \|x -y \| \over \|x -y \|  \|x - \hat{y} \| }  \right | \|x-y\| + {1 \over \|x -\hat{y}  \| }  \|y -\hat{y}    \| \\
	&\le \left | 	{ \|y - \hat{y}  \| \over  \|x - \hat{y} \| }  \right | + {1 \over \|x -\hat{y}  \| }  \|y -\hat{y}  \| \\
	&=  {2\over \|x -\hat{y}  \| }  \|y -\hat{y}   \| \le {2 \delta \over r}. 
	\end{align*}
\end{proof}

\begin{lem}
	\label{lem:stable_perturbation_lemma1}
	Let $z, \hat{z}$ be $s$-sparse vectors with $\|z - \hat{z} \|_2 \le \delta$. Then, we have
	\[
	|\inner{a_i,z} -\inner{a_i,\hat{z} } | \le \delta \left(C_b  \sqrt{ s\log(N/s) } + \sqrt{ {5 \log m \over c_b}}\right)
	\]
	for all $1 \le i \le m$ with probability at least $1 - {2 \over m^4}$. 
\end{lem}

\begin{proof}
	Let $h = z - \hat{z}$. By Proposition 	\ref{prop:Uniform_bound_for_dual_norm_Gaussian_vector}, for all $i \in [m]$, we have
	\begin{align}
	\label{eq: Uniform_bound_for_difference}
	 \sup\limits_{h \in K-K \cap \delta \mathbb{S}^{N-1}}  \left| \inner{a_i,h } \right| = \|a_i\|_{K_\delta^\circ} = \delta \|a_i\|_{K_1^\circ} 
	\le \delta \left(C_b  \sqrt{ s\log(N/s) } + \sqrt{ {5 \log m \over c_b}}\right)
	\end{align} with probability at least $1 - {2 \over m^5}$.  Then the lemma follows from the union bound. 
\end{proof}

\begin{lem}
	\label{lem:stable_perturbation_lemma2}
	Let $z, \hat{z}$ be unit $s$-sparse vectors with $\|z - \hat{z} \|_2 \le \delta$. Then, for $u > 0$, we have
	\[
	\| \inner{a_i,z}z -\inner{a_i,\hat{z} } \hat{z} \|_{K_1^\circ}  \le 2\delta \left(C_b  \sqrt{ s\log(N/s) } + \sqrt{ {5 \log m \over c_b}}\right)
	\]
	for all $1 \le i \le m$ with probability at least $1 - {2 \over m^4}$. 
\end{lem}

\begin{proof}
	Note that 
	\begin{align*} 
		\| \inner{a_i,z}z -\inner{a_i,\hat{z} } \hat{z} \|_{K_1^\circ} 
		& \le  \| \inner{a_i,z}z -\inner{a_i,\hat{z} } z \|_{K_1^\circ} + \| \inner{a_i,\hat{z} } z -\inner{a_i,\hat{z} } \hat{z} \|_{K_1^\circ}\\
		& \le | \inner{a_i,z}-\inner{a_i,\hat{z} }| \cdot \| z \|_{K_1^\circ} + |\inner{a_i,\hat{z} } | \cdot \| z -  \hat{z} \|_{K_1^\circ}\\
		& \le | \inner{a_i,z}-\inner{a_i,\hat{z} }| \cdot \|z\|_2 + |\inner{a_i,\hat{z} } | \cdot \| z -  \hat{z} \|_2\\
		& \le  \delta \left(C_b  \sqrt{ s\log(N/s) } + \sqrt{ {5 \log m \over c_b}}\right) + \left(C_b  \sqrt{ s\log(N/s) } + \sqrt{ {5 \log m \over c_b}}\right) \cdot \delta \\
		& \le  2\delta \left(C_b  \sqrt{ s\log(N/s) } + \sqrt{ {5 \log m \over c_b}}\right).
	\end{align*}
	The explanations to above inequalities are as follows: We used the triangle inequality in the first and second inequalities. The third inequalities is from the definition of the dual norm $\| \cdot \|_{K_1^\circ}$. We applied Lemma  \ref{lem:stable_perturbation_lemma1} to the fourth inequality. 	
\end{proof}

Lastly, define a quantity $C(N,s,m)$ to be 
\begin{align}
\label{eq:definition_of_C(N,s,m)}
C(N,s,m) :=    3\left(C_b  \sqrt{ s\log(N/s) } + \sqrt{ {5 \log m \over c_b}}\right).
\end{align}
Note also that $C(N,s,m) \le 9 C_b  \sqrt{ s\log(N/s) } \cdot  \sqrt{ {5 \log m \over c_b}}$. 

 Using this notation, we summarize Lemma \ref{lem:Uniform_bound 1}, \ref{lem:stable_perturbation_lemma1}, and \ref{lem:stable_perturbation_lemma2} in one lemma below.

\begin{lem}
	\label{lem: Unified uniform bound}
	With probability at least $1 - {2 \over m^4}$, for all $i$ with $1 \le i \le m$, we have
 	\[
		\sup\limits_{y \in K \cap \mathbb{S}^{N-1}} \left \{ \left| \inner{a_i,{x +y \over \|x +y \|}  } \right|,  \left| \inner{a_i,{x -y \over \|x -y \|}  } \right|, \|b_i(x,y)\|_{K_1^\circ} \right \} \le  C(N,s,m)
	\] and
	\[
	\max \left \{|\inner{a_i,z} -\inner{a_i,\hat{z} } |,\| \inner{a_i,z}z -\inner{a_i,\hat{z} } \hat{z} \|_{K_1^\circ} \right  \} \le 2\delta \cdot C(N,s,m)
	\] for any $s$-sparse unit vectors $z, \hat{z}$ with $\|z - \hat{z} \|_2 \le \delta$.
\end{lem}

\subsection{Local binary embedding for small regions at \texorpdfstring{$x$}{Lg}}
\label{section:local_embedding}
One of key ingredients of our proof of the RAIC is the local binary embedding (local sensitivity hashing) property by Oymak and Recht \cite{oymak2015near}.  
Bilyk and Lacey \cite{bilyk2015random} also reported a similar property.

\begin{thm} [Local $\delta$-binary embedding \cite{oymak2015near,bilyk2015random}] 
	\label{thm:Local binary embedding}
	Let $A \in \R^{m \times N}$ be a standard Gaussian random matrix. Then, there exists  universal constants $C_l, C_L \ge 1$, and $\tilde{c}>0$ such that given a set $G \in \mathbb{S}^{N-1}$ and a constant $\delta \in (0,1)$, we have 
	\begin{enumerate}
		\item For all $x, y \in G$ with $\|x-y\|_2 \le \delta/\sqrt{\log \delta^{-1}}$, $\left| {1 \over 2m}\|\text{sign}(Ax) - \text{sign}(Ay)\|_1 - d_g(x, y) \right| \le C_L \delta$,
		\item Conversely for all  $x, y \in G$ with $\left| {1 \over 2m}\|\text{sign}(Ax) - \text{sign}(Ay)\|_1 - d_g(x, y) \right| \le C_L \delta$, $\|x-y\|_2 \le \delta$,
	\end{enumerate}
	with probability $1-\exp(-\tilde{c} \delta m)$ whenever $m \ge C_l\delta^{-3} \log \delta^{-1}  w^2(G)$. Here $w(G)$ is the Gaussian width of $G$. 
	
\end{thm}

Define a constant $C_{10}$ as $C_{10} := \max \{C_l, C_L, 2C^2_b\} + \pi$. We are looking for two sequences $\{r_i\}_{i=1}$ and $\{\delta_i\}_{i=1}$ satisfying the following properties. 
\begin{enumerate}
	\item $r_1 = m^{-1/2}C(N,s,m).$
	
	\item $\delta_i =C(N,s,m) \left({r^2_i\log m \over m}\right)^{1/3} \log m$.
	
	\item $r^2_{i+1} = 600 C_{10} \log m \cdot  r_i \delta_i C(N,s,m)  \quad \text{or} \quad r_{i+1} = 10 \sqrt{6} C_{10}^{1/2}  r_i^{5/6} m^{-1/6} (\log m)^{7/6} C(N,s,m)$. 
\end{enumerate}

\begin{prop}
	\label{prop: Properties_of_sequences}
	It is easy to see that $\delta_i \ge {1 \over m}$ and $r_i \ge {1 \over m}$ for all $i$, so $\log m \ge \log \delta_i^{-1}$ and $\log m \ge \log r_i^{-1}$. Also note that
	${r_{i+1} \over 600} = C_{10} {r_i \delta_i \log m \cdot C(N,s,m) \over r_{i+1}}$ and ${r_{i+1} \over 10 \sqrt{6}} = \left(C_{10} r_i \delta_i \log m \cdot C(N,s,m)\right)^{1/2}$.
\end{prop}

The following two sequences $\{r_i\}$ and $\{\delta_i\}$ are constructed by induction based on above three requirements for the sequences above.
\begin{dfn} For $i \ge 0$, 
	\label{def:Sequences_for_nets}
	\begin{align*}
	& r_{i+1} = (600C_{10})^{3 \left(1 -\left({5 \over 6} \right)^{i} \right)} r_{1}^{({5 \over 6})^{i}}m^{- \left(1 -\left({5 \over 6} \right)^{i} \right)} (\log m)^{ 7\left(1 -\left({5 \over 6} \right)^{i} \right)} C(N,s,m)^{6 \left(1 -\left({5 \over 6} \right)^i \right)}, \\  
	& \qquad = (600 C_{10})^{3 \left(1 -\left({5 \over 6} \right)^{i} \right)} m^{- \left(1 -{1 \over 2}\left({5 \over 6} \right)^{i} \right)}  (\log m)^{7 \left(1 -  \left({5 \over 6} \right)^{i} \right)} C(N,s,m)^{6 - 5\left({5 \over 6} \right)^i}, \\  
	& \delta_{i+1} = (600C_{10})^{2 \left(1 -\left({5 \over 6} \right)^{i} \right)}  m^{- \left(1 -{1 \over 3}\left({5 \over 6} \right)^{i} \right)} (\log m)^{ {14 \over 3}\left(1 -\left({5 \over 6} \right)^{i} \right) +{4 \over 3}} C(N,s,m)^{5  - {10 \over 3}\left({5 \over 6} \right)^i}.   
	\end{align*}
\end{dfn}

\begin{prop}
\label{prop: Properties_of_sequences2}
    From the definitions of $r_{i}$, $r_{i+1}$, and $C(N,s,m)$, it is straightforward to check that there exists $c > 0$ such that $r_{i+1} \le r_i$, ${r_{i+1} \over 600} \le C_{10} \delta_i \log m \cdot C(N,s,m) $, and $\delta_i \le {r_{i+1} \over 600}$ as long as $m > c s^{10} \log^{10}(N/s)$.
\end{prop}

Let $K^{(i)} := B(x, r_i) \cap \Sigma^N_{s}$.
Then, for $y \in K^{(i)}$,  $y-x \in B(0, r_i) \cap \Sigma^N_{2s}$, so $y \in B(0, r_i) \cap \Sigma^N_{2s} + x$, i.e., $K^{(i)} \subset B(0, r_i) \cap \Sigma^N_{2s} + x$. Hence, we have
\begin{align*}
w(K^{(i)}) &= w \left( B(x, r_i) \cap \Sigma^N_{s} \right) \\
&\le w \left(B(0, r_i) \cap \Sigma^N_{2s} + x \right) \quad \text{( for any sets $S, T$ with $S \subseteq T, w(S) \le w(T)$ ) } \\ 
&= w \left(B(0, r_i) \cap \Sigma^N_{2s} \right) \quad \text{( for any set $S, w(S+x) = w(S)$ ) }\\
&\le r_i \cdot C_b \sqrt{2s \log (N/s)}. 
\end{align*}

Next, we apply the local binary embedding to the set $K^{(i)}$ as follows.

\begin{cor} [Corollary of local $\delta$-binary embedding] 
	\label{cor:Local binary embedding}
	Under the same notations in Theorem \ref{thm:Local binary embedding}, we have
	\begin{enumerate}
		\item For all $y \in K^{(i)}$ with $\|x-y\|_2 \le \delta_i$, $\left| {1 \over 2m}\|\text{sign}(Ax) - \text{sign}(Ay)\|_1 \right| \le C_{10} \delta_i \log m$,
		
		\item For all $y \in K^{(i)}$ with $\|x-y\|_2 \le r_i$, $\left| {1 \over 2m}\|\text{sign}(Ax) - \text{sign}(Ay)\|_1 \right| \le C_{10} r_i \log m$,
	\end{enumerate}
	with probability $1 - {c_2 \over m^5}$ for some universal constant $c_2 > 0$.
\end{cor}

\begin{proof}
     First, choose $\delta$ such that $\delta_i = \delta/\sqrt{\log \delta^{-1}}$. Since $y \in K^{(i)}$ with $\|x-y\|_2 \le \delta_i$ and $\log m \ge \log \delta^{-1}_i \ge  \log \delta^{-1}$ in Proposition \ref{prop: Properties_of_sequences}, the first part of Theorem \ref{thm:Local binary embedding} implies that 
     \[
     \left| {1 \over 2m}\|\text{sign}(Ax) - \text{sign}(Ay) \|_1 - d_g(x, y) \right| \le C_L \delta_i \log m.
     \]
     %with with probability exceeding $1 - \exp(-\tilde{c} m \delta )$ for $1-\exp(-\tilde{c} \delta m)$.
     
     Next, note that $d_g(x,y) \le \pi \|x-y\|_2$ since $x$, $y$ are unit vectors, so $d_g(x,y) \le \pi \delta_i$ from the relation between the arc and chord lengths in the unit sphere.  
	 Thus, from the fact that $C_{10} \ge C_L + \pi$,  we have $\left| {1 \over 2m}\|\text{sign}(Ax) - \text{sign}(Ay)\|_1 \right| \le C_{10} \delta_i \log m$  with probability at least 
	 \begin{align*}
	    & 1 - \exp(-\tilde{c}\delta m  ) \\
	    & \ge 1 - \exp(-\tilde{c} \delta_i  m ) \\
	 	& \ge 1-\exp \left(- \tilde{c} (600C_{10})^{2 \left(1 -\left({5 \over 6} \right)^{i-1} \right)}  m^{{1 \over 3}\left({5 \over 6} \right)^{i-1}}  (\log m)^{ {14 \over 3}\left(1 -\left({5 \over 6} \right)^{i-1} \right) +{4 \over 3}} C(N,s,m)^{5  - {10 \over 3}\left({5 \over 6} \right)^{i-1}}  \right) \\
	 	&\ge 1-\exp \left(- \tilde{c} (\log m)^2 \right) \ge 1 - {c_2 \over m^5} \qquad \text{for some absolute constant $c_2 > 0$,}
	 \end{align*} 
	if the condition $m \ge C_l\delta^{-3} \log \delta^{-1}  w^2(K^{(i)})$ is met. By the construction of $\delta_i$ and $C(N,s,m) $, we have  $m \ge  {\left(\delta_i \over \log m \right)}^{-3}  r^2_i  \log m  \cdot C^3(N,s,m) \ge C_l\delta^{-3} \log \delta^{-1}  w^2(K^{(i)})$, so  this condition is satisfied. 
	 This proves the first part of the corollary and the second part follows from the same arguments. 
	 
\end{proof}

\subsection{\texorpdfstring{ $\epsilon$-net for  $K^{(i)}$}{Lg}. }
\label{section:Epsilon_net}
Consider an $\epsilon$-net $\mathcal{N}_{K^{(i)}}$ for $K^{(i)}$ with $\epsilon = \delta_{i}$. Then, by the Sudakov minorization inequality [Theorem 7.4.1 in Vershynin \cite{vershynin2018high}], there exist constants $C', C^{''} > 0$ such that for all $i \ge 1$,
\begin{align*}
\log | \mathcal{N}_{K^{(i)}}| 
& \le {C' w^2(K^{(i)}) \over \delta_{i}^2} \\
& \le {C' r_i^2 w^2(K) \over \delta_{i}^2} \\
& \le  C'(600C_{10})^{2 \left(1 -\left({5 \over 6} \right)^{i-1} \right)} m^{{1 \over 3}({5 \over 6})^{i-1}} (\log m)^{2 \left( 1- {7 \over 3}\left({5 \over 6} \right)^{i-1} \right)} C(N,s,m)^{2  - {10 \over 3}\left({5 \over 6} \right)^{i-1}}w^2(K)\\
& \le  C^{''}(600C_{10})^{2 \left(1 -\left({5 \over 6} \right)^{i-1} \right)} m^{{1 \over 3}({5 \over 6})^{i-1}} (\log m)^{2 \left( 1- {7 \over 3}\left({5 \over 6} \right)^{i-1} \right)} C(N,s,m)^{4  - {10 \over 3}\left({5 \over 6} \right)^{i-1}}.
\end{align*}

\subsection{\texorpdfstring{ Correlation between $\text{sign}\inner{a,x}$ and $\inner{a,y}$}{Lg} }

\begin{prop}
	\label{prop:one-bit correlation with Gaussian}
	Let $x, y$ be unit vectors and $a$ be a Gaussian random vector in $\R^N$. Then, 
	\[
	\sqrt{\pi \over 2} \mathbb{E} [ (\text{sign}\inner{a,x}) \inner{a,y} ] = \inner{x,y} 
	\] and
	\[
	\sqrt{\pi \over 2} \mathbb{E} \Big[{1 \over m} A^T \text{sign}(Ay) \Big] = y.
	\]
	
\end{prop}
\begin{proof}
See the proof of Lemma 4.1 in \cite{plan2012robust}.

\iffalse
	From Lemma 8.1 in \cite{plan2016high}, $\inner{a,x}$ and $a - \inner{a,x}x$ are independent. Since $\text{sign}\inner{a,x}$ is a function of $\inner{a,x}$, it is independent with $\inner{a - \inner{a,x}x, y}$ as it is also a function of $a - \inner{a,x}x$. Thus, 
	\begin{align*}
	 \sqrt{\pi \over 2} \mathbb{E} [ (\text{sign}\inner{a,x}) \inner{a,y} ] 
		 & = \sqrt{\pi \over 2} \mathbb{E} [ (\text{sign}\inner{a,x}) \inner{\inner{a,x}x + a - \inner{a,x}x,y} ] \\
		&  = \sqrt{\pi \over 2} \mathbb{E} [ (\text{sign}\inner{a,x}) \inner{a,x} ]\inner{x,y} + \sqrt{\pi \over 2} \mathbb{E} [ (\text{sign}\inner{a,x}) \inner{ a - \inner{a,x}x,y} ] \\
		&  = \sqrt{\pi \over 2} \mathbb{E} [ (\text{sign}\inner{a,x}) \inner{a,x} ]\inner{x,y} + \sqrt{\pi \over 2} \mathbb{E} [ (\text{sign}\inner{a,x}) ]  \mathbb{E} [ \inner{ a - \inner{a,x}x,y} ] \\
		&  = \sqrt{\pi \over 2} \mathbb{E} [ |\inner{a,x}| ]\inner{x,y} + 0\\
		& = \inner{x,y} + 0 = \inner{x,y}.
	\end{align*}
\fi
\end{proof}

\section{Main technical propositions}
\label{section:Main technical propositions}
This section proves our main theorem and technical propositions. Our goal is to establish the RAIC for a certain small region around $x$ which is crucial in our analysis. 
With all the facts gathered in the previous section, we are prepared to state our main technical proposition for RAIC precisely.

\begin{prop}
	\label{prop:Main_proposition} 
	Let $x$ be a $s$-sparse unit vector in $\R^N$. Then, the following bound holds uniformly for any $s$-sparse unit vector $y$ with $r_{i+1} \le \|x-y\| \le r_i$.
	\begin{align*}
	\left \| {\tau \over m}\cdot A^T(\text{sign}(Ax) - \text{sign}(Ay)) -  (x- y) \right \|_{K_1^\circ} \le  {3 \over m^{{1 \over 40} ({5 \over 6})^{i}}} \|x - y\|_2 + {3 \over 50}r_{i+1} 
	\end{align*} with probability exceeding 
	\begin{align*}
	& 1 - 3\bar{c}  \left({s \over N} \right)^{5 s} \cdot {1 \over m^5}  -  {c_9 \over m^4}
	\end{align*}
	for some universal constants $\bar{c}, c_9 > 0$.
	
	In other words, the measurement matrix $A$ satisfies  $\left({\sqrt{2  \over \pi}\cdot {1 \over m}}, {3 \over m^{{1 \over 40} ({5 \over 6})^i }},{3 \over 50}r_{i+1},  r_{i+1}, r_i  \right)$-RAIC with high probability. 
\end{prop}

\subsection{Orthogonal decomposition of measurement vectors}
The first step of the proof of Proposition \ref{prop:Main_proposition} is the decomposition of ${\tau \over m} A^T(\text{sign}(Ax) - \text{sign}(Ay)) -(x-y)$.
Essentially, it decomposes into three parts: the components along the direction $x-y$, $x+y$, and their orthogonal part.
This decomposition is based on Lemma \ref{lem:Orthogonal_decomposition_Gaussian_vector}.

\begin{proof} [Proof of Proposition \ref{prop:Main_proposition} ]
Using the expansion of $a_i$ in Lemma \ref{lem:Orthogonal_decomposition_Gaussian_vector} and the triangle inequality, we have
\begin{align*}
&\left \| {\tau \over m}\cdot A^T(\text{sign}(Ax) - \text{sign}(Ay)) -  (x- y) \right \|_{K_1^\circ} \\
&= \left \| {\tau \over m} \sum_{i=1}^m(\text{sign} (\inner{a_i,x}) - \text{sign}(\inner{a_i, y})a_i -  (x- y) \right \|_{K_1^\circ} \\
&  \le  \left \| {\tau \over m} \sum_{i=1}^m(\text{sign} (\inner{a_i,x}) - \text{sign}(\inner{a_i, y})\left(\inner{a_i, {x -y \over \|x -y \| }} {x -y \over \|x -y \| } + \inner{a_i,{x + y\ \over \|x +y \|}  }{x +y \over \|x +y \| } \right) -  (x- y) \right \|_{K_1^\circ} \\
&\quad + \left \| {\tau \over m} \sum_{i=1}^m(\text{sign} (\inner{a_i,x}) - \text{sign}(\inner{a_i, y}) b_i(x,y) \right \|_{K_1^\circ} \\
&\le  \left \| {\tau \over m} \sum_{i=1}^m(\text{sign} (\inner{a_i,x}) - \text{sign}(\inner{a_i, y}) \inner{a_i, {x -y \over \|x -y \| }} {x -y \over \|x -y \| }   -  (x- y) \right \|_{K_1^\circ} \\
&\quad + \left \| {\tau \over m} \sum_{i=1}^m(\text{sign} (\inner{a_i,x}) - \text{sign}(\inner{a_i, y}) \inner{a_i,{x +y \over \|x +y \|}  }{x +y \over \|x +y \| }  \right \|_{K_1^\circ} \\
&\quad + \left \| {\tau \over m} \sum_{i=1}^m(\text{sign} (\inner{a_i,x}) - \text{sign}(\inner{a_i, y}) b_i(x,y) \right \|_{K_1^\circ} \\
&\le  \left | {1 \over m} \sum_{i=1}^m \tau (\text{sign} (\inner{a_i,x}) - \text{sign}(\inner{a_i, y}) \inner{a_i, {x -y \over \|x -y \| }}  -  \|x- y\| \right| \left \| {x -y \over \|x -y \| }   \right \|_{K_1^\circ} \\
&\quad + \left | {1 \over m} \sum_{i=1}^m \tau(\text{sign} (\inner{a_i,x}) - \text{sign}(\inner{a_i, y}) \inner{a_i,{x +y \over \|x +y \|}  } \right|  \left \| {x +y \over \|x +y \| }  \right \|_{K_1^\circ} \\
&\quad + \left \| {\tau \over m} \sum_{i=1}^m(\text{sign} (\inner{a_i,x}) - \text{sign}(\inner{a_i, y}) b_i(x,y) \right \|_{K_1^\circ} \\
& =  \underbrace{\left | {1 \over m} \sum_{i=1}^m \tau (\text{sign} (\inner{a_i,x}) - \text{sign}(\inner{a_i, y}) \inner{a_i, {x -y \over \|x -y \| }}  -  \|x- y\| \right| }_{(I)}     \\
&\quad +  \underbrace{ \left | {1 \over m} \sum_{i=1}^m \tau(\text{sign} (\inner{a_i,x}) - \text{sign}(\inner{a_i, y}) \inner{a_i,{x +y \over \|x +y \|}  } \right| }_{(II)}  \\
&\quad + \underbrace{  \left \| {\tau \over m} \sum_{i=1}^m(\text{sign} (\inner{a_i,x}) - \text{sign}(\inner{a_i, y}) b_i(x,y) \right \|_{K_1^\circ}  }_{(III)} .
\end{align*}
The above inequality brings to control the three terms $(I), (II)$, and $(III)$, which is the main technical challenge of our work. The subsequent three lemmas provide upper bounds for these terms, which will be proven in the next section. 

\begin{lem}
	\label{lem:Bound_for_(I)}
	There exist universal constants $c, \bar{c} > 0$ such that for all $s$-sparse unit vector $y$ with $r_i \le \|x -y \|_2 \le r_{i+1}$, we have
	\[
		(I) \le {1 \over m^{{1 \over 40}({5 \over 6})^i}} \|x -y\|_2  + {r_{i+1} \over 50} 
	\] for all $m > c s^{10} \log^{10}(N/s)$
	with probability at least  $1 -   \bar{c}  \left({s \over N} \right)^{5 s} \cdot {1 \over m^5}  - {c_6 \over m^4}$ 

\end{lem}

\begin{lem}
	\label{lem:Bound_for_(II)}
	There exist universal constants $c, \bar{c} > 0$ such that for all $s$-sparse unit vector $y$ with $r_i \le \|x -y \|_2 \le r_{i+1}$,
	\[
	(II) \le {2 \over m^{{1 \over 40}({5 \over 6})^i}} \|x -y\|_2  + {r_{i+1} \over 50} 
	\] for all $m > c s^{10} \log^{10}(N/s)$
	with probability at least  $1 -   \bar{c}   \left({s \over N} \right)^{5 s} \cdot {1 \over m^5}  - {c_6 \over m^4}$.

\end{lem}

\begin{lem}
	\label{lem:Bound_for_(III)}
	There exist universal constants $c, \bar{c} > 0$ such that for  all $s$-sparse unit vector $y$ with $r_i \le \|x -y \|_2 \le r_{i+1}$,
	\[
	(III) \le {r_{i+1} \over 50} 
	\] for all $m > c s^{10} \log^{10}(N/s)$
	with probability at least $ 1 -   \bar{c}  \left({s \over N} \right)^{5 s} \cdot {1 \over m^5}- {c_8 \over m^4}$.
\end{lem}

Applying Lemma \ref{lem:Bound_for_(I)}, \ref{lem:Bound_for_(II)}, and \ref{lem:Bound_for_(III)} to terms (I), (II), (III) and setting the constant $c_9 := 2c_6 + c_8$ complete the proof of Proposition \ref{prop:Main_proposition}. 

\end{proof}

Proposition \ref{prop:Main_proposition} leads to our main global convergence theorems. 

\begin{thm}
	\label{thm: Contraction_with_additivie_error}
	If $x_k$, the $k$-th iterate of NBHIT satisfies $r_{i+1} \le \|x_k - x\| \le r_i$ for some $i$, then we have
	\[
	\|x_{k+1} - x\|_2 \le {12 \over m^{{1 \over 40} ({5 \over 6})^i}} \|x_k - x\|_2 + {6 \over 25}r_{i+1} 
	\] with probability at least
	\begin{align*}
	& 1 -  3 \bar{c}  \left({s \over N} \right)^{5 s} \cdot {1 \over m^5}  -  {c_9 \over m^4}.
	\end{align*}
\end{thm}

\begin{proof} [Proof of Theorem \ref{thm: Contraction_with_additivie_error}] \mbox{}
	%The iterate $x_1$ satisfies $\|x_1 - x\|_2 \le {4 \eta \over \sqrt{m}} \left[s + \sqrt{ s \log (2n/s)} \right]$ by Section \ref{section:FirstIteration} with high probability. 
	
	Applying Proposition \ref{prop:Main_proposition} to $x_k$ and using the relation $z_{k+1} = x_k + {\tau \over m} A^T(\text{sign}(Ax) - \text{sign}(Ax_k))$ yield
	\[
	\|z_{k+1} - x\|_{K_1^\circ}  \le   {3 \over m^{{1 \over 40} ({5 \over 6})^{i}}} \|x - y\|_2 + {3 \over 50}r_{i+1}.
	\]
	
	Then, from \eqref{ineq:Main_inequality}, 
	\begin{align*}
	\|x_{k+1} -x \|_2  & \le 4\left(  {3 \over m^{{1 \over 40} ({5 \over 6})^{i}}} \|x - y\|_2 + {3 \over 50}r_{i+1} \right)  \\
	& \le  {12 \over m^{{1 \over 40} ({5 \over 6})^i}} \|x_k - x\|_2 + {6 \over 25}r_{i+1}
	\end{align*} 
\end{proof}

\begin{cor}
	\label{cor:Main_gloabl_convergence}
	Suppose $m^{{1 \over 40} ({5 \over 6})^L} > 24$ for some integer $L > 0$. Let $i$ be an integer with $1 \le i \le L$ and $x_k$ satisfy $r_{i+1} \le \|x_k - x\| \le r_i$.
	Then, we have for any $t \ge 1$, either there exists $p$ with $1 \le p \le t$ such that $\|x_{k+p} -x\| \le r_{i+1}$ or 
	\[
	\|x_{k+t} - x\|_2 \le   {1 \over 2^t} \cdot {1 \over m^{{t \over 40} \left(({5 \over 6})^i- ({5 \over 6})^L \right) } } \|x_k - x\|_2 + {12 \over 25}r_{i+1} .
	\] 
\end{cor}

\begin{proof}
	Since $m^{{1 \over 40} ({5 \over 6})^L} > 24$, $ {12 \over m^{{1 \over 40} ({5 \over 6})^i}} < {1 \over 2} \cdot {1 \over m^{{1 \over 40} \left(({5 \over 6})^i- ({5 \over 6})^L \right) } }$. Thus, from Theorem \ref{thm: Contraction_with_additivie_error}
	\[
	\|x_{k+1} - x\|_2 \le {1 \over 2} \cdot {1 \over m^{{1 \over 40} \left(({5 \over 6})^i- ({5 \over 6})^L \right) } } \|x_k - x\|_2 +  {6 \over 25}r_{i+1}.
	\] 
	Note that ${1 \over 2} \|x_k - x\|_2 +  {6 \over 25} r_{i+1} \le {1 \over 2} r_i + {1 \over 2} r_i \le r_i$, so $\|x_{k+1} -x\|_2 \le r_i$. 
	
	After noticing ${1 \over 2} \cdot {1 \over m^{{1 \over 40} \left(({5 \over 6})^i- ({5 \over 6})^L \right) } } \le {1 \over 2}$, one can show that the following by induction: For any $t \ge 1$, either there exists $p$ with $1 \le p \le t$ such that $\|x_{k+p} -x\| \le r_{i+1}$ or 
	\[
	\|x_{k+t} - x\|_2 \le   {1 \over 2^t} \cdot {1 \over m^{{t \over 40} \left(({5 \over 6})^i- ({5 \over 6})^L \right) } } \|x_k - x\|_2 +  {12 \over 25} r_{i+1}
	\] with high probability.
	
\end{proof}

From the relation between $r_i$ and $r_{i+1}$, it turns out that $\|x_{k+p} -x\| \le r_{i+1}$ for some $p$ with $1 \le p \le 25$ whenever $x_k \in K^{i}$. In other words, $25$ iterations are enough for the BIHT iterates to reach next level $K^{(i+1)}$. This is basically the idea of Theorem \ref{thm: Main Convergence Thm}, which we present in the following corollary. 

\begin{cor}
	\label{cor:Main_gloabl_convergence2}
	Suppose $m^{{1 \over 40} ({5 \over 6})^L} > 24$  for some positive integer $L$. 
	Then, for any integer $i$ with $0 \le i < L$, if $\|x_k - x\| \le r_i$, there exists $t \le 25$ such that 
	\[
	\|x_{k+j} - x\|_2 \le r_i \quad \text{for all $0 \le j \le t$}
	\]  and
	\[
	\|x_{k+t} - x\|_2 \le r_{i+1}.
	\]
\end{cor}

\begin{proof}
	Suppose the claim in the Corollary is not true. Then, there exists $t > 25$ such that 
	$\|x_{k+p} - x\| > r_{i+1}$ for all $p$ with $1 \le p \le t$. Also, $\|x_{k+p} - x\| \le r_i$ as in the proof in Corollary \ref{cor:Main_gloabl_convergence}. 
	On the other hand, it is easy to check that
	\[
	t + t \left(1 - \left({5 \over 6}\right)^{L-i} \right) \cdot  {1 \over 40} \left({5 \over 6}\right)^i \log_2 m \ge \log_2 {25} + {1 \over 10} \left( {5 \over 6}\right)^i \log_2 m \quad \text{for $t > 25$.}
	\] 
	Then, from Definition \ref{def:Sequences_for_nets} for $r_i, r_{i+1}$, this implies that $ {1 \over 2^t} \cdot {1 \over m^{{t \over 40} \left(({5 \over 6})^i- ({5 \over 6})^L \right) } }  r_i \le {1 \over 25}r_{i+1}$ by taking the logarithm to the base 2 on this equality. However, then we have $\|x_{k+25} - x\| \le {1 \over 25}r_{i+1} + {12 \over 25}r_{i+1} \le r_{i+1}$ by Corollary \ref{cor:Main_gloabl_convergence}, which contradicts to the assumption that $\|x_{k+25} - x\| > r_{i+1}$.
\end{proof}

We restate Theorem \ref{thm: Main Convergence Thm} for the convenience of readers. 
\begin{thm}
	\label{thm:Main_gloabl_convergence3}
	Suppose $m > \max \{c s^{10} \log^{10} (N/s), 24^{48} \}$ for a universal constant $c > 0$. Then, there exists a universal constant $C_{12} > 0$ such that the NBIHT iterates obey
	\[
	\|x_{k} - x\|_2 \le   C_{12}{ (s \log(N/s) )^{7/2}(\log m)^{12}  \over m^{ \left(1 -{1 \over 2}\left( {5 \over 6} \right)^{  \lfloor{k/25}\rfloor-1 }  \right) } }
	\]  for $k \ge 1$ with probability exceeding $1 - O \left( {1 \over m^3} \right)$.
\end{thm}

\begin{proof} [Proof of Theorem \ref{thm:Main_gloabl_convergence3} ] \mbox{}
	Our proof is based on Corollary \ref{cor:Main_gloabl_convergence2} and a similar type of argument used to show the stability of the Truncated Wirtinger Flow \cite{chen2017solving}.
	
	Let $L$ be the largest positive integer satisfying $m^{{1 \over 40} ({5 \over 6})^L} > 24$. Note that such $L$ always exists from the assumption $m > 24^{48} = 24^{40\left({6 \over 5}\right)}$. We shall call $\{y: \|y -x \| > r_L \}$ Regime $A$ and $\{y: \|y -x \| \le r_L \}$ Regime $B$. 
	
	First, for $x_k$ belongs to Regime $A$, applying Corollary \ref{cor:Main_gloabl_convergence2} repeatedly shows that there exists $k \le 25i$ such that $\|x_k - x\| \le r_{i}$ as long as $i \le L$. 
	
	Next, we will prove that if $x_k$ is in Regime $B$ ($\|x_k - x\| \le r_{L}$), then $\|x_{k+1} - x\| \le 8 C_{10}  r_{L} \log m C(N,s,m)$. In other words, for  the approximation error of $x_{k+1}$ is still well-controlled for $x_k$ belonging to Regime $B$. To see this, we start from \eqref{ineq:Main_inequality}, 
	\begin{align*}
	\left\| x_{k+1}- x \right \|_2
	& \le 4 \left \| {\tau \over m}\cdot A^T(\text{sign}(Ax) - \text{sign}(Ax_k)) -  (x- x_k) \right \|_{K_1^\circ} \\
	& \le 4 \left \| {\tau \over m}\cdot A^T(\text{sign}(Ax) - \text{sign}(Ax_k))  \right \|_{K_1^\circ} +  4\|  x- x_k \|_{K_1^\circ} \\
	& \le 4 \left \| {\tau \over m} \sum_{i=1}^m(\text{sign} (\inner{a_i,x}) - \text{sign}(\inner{a_i, x_k}) a_i \right \|_{K_1^\circ} +  4\|  x- x_k \|_{K_1^\circ} \\
	& \le  4\cdot {\tau \over m} \sum_{i=1}^m |\text{sign} (\inner{a_i,x}) - \text{sign}(\inner{a_i, x_k})|  \left \| a_i \right \|_{K_1^\circ} +  4\|  x- x_k \|_{K_1^\circ} \\
	& \le  4 \max_{1 \le j \le m} {\left \| a_j \right \|_{K_1^\circ} } \cdot {\tau \over m} \sum_{i=1}^m |\text{sign} (\inner{a_i,x}) - \text{sign}(\inner{a_i, x_k})|  +  4\|  x- x_k \|_{K_1^\circ} \\
	& \le 4C_{10} r_{L} \log m \cdot C(N,s,m) + 4 \|x- x_k\|_2\\
	& \le 4 C_{10} r_{L} \log m \cdot C(N,s,m) + 4 r_{L} \\
	& \le 8  C_{10}  r_{L} \log m \cdot C(N,s,m)
	\end{align*} where the first five inequalities follow from the triangle inequality. The fifth inequality arises from Corollary \ref{cor:Local binary embedding}, the definition of the event $E_u$, and the definition of $C(N,s,m)$. 
	
	If $x_{k+1}$ lands on the region $A$ (i.e., $\|x_{k+1} -x \| > r_L$), then we can again apply Corollary \ref{cor:Main_gloabl_convergence2}. So the rest of iterations are guaranteed to satisfy
	\[
	\left\| x_{k+j}- x \right \|_2 \le 8 C_{10}  \log m \cdot C(N,s,m)  r_{L}
	\] for all $j \ge 1$. 
	Hence, above argument, Definition \ref{def:Sequences_for_nets}, and \eqref{eq:definition_of_C(N,s,m)} 
	yield
	\[
	\|x_k - x\|_2 \le   8 C_{10}   \log m \cdot C(N,s,m) \cdot \max \{ r_{ \lfloor{k/25}\rfloor },r_{L}  \} \le   C_{11}{ (s \log(N/s) )^{7/2}(\log m)^{12}  \over m^{ \left(1 -{1 \over 2}\left( {5 \over 6} \right)^{ \min\{ \lfloor{k/25}\rfloor, L \}-1 }  \right) } }
	\] for some universal constant $C_{11}$.
	
	Now, note that the above expression for the error bound can be rewritten as
	\begin{align}
	 C_{11}{ (s \log(N/s) )^{7/2}(\log m)^{12}  \over m^{ \left(1 -{1 \over 2}\left( {5 \over 6} \right)^{ \min\{ \lfloor{k/25}\rfloor, L \}-1 }  \right) } } \nonumber
	 & = C_{11}{ (s \log(N/s) )^{7/2}(\log m)^{12} m^{ {1 \over 2}\left( {5 \over 6} \right)^{ \min\{ \lfloor{k/25}\rfloor, L \}-1 }  } \over m } \nonumber \\
	 & = C_{11}{ (s \log(N/s) )^{7/2}(\log m)^{12} \max\{ m^{ {1 \over 2}\left( {5 \over 6} \right)^{  \lfloor{k/25}\rfloor -1 }  }, m^{ {1 \over 2}\left( {5 \over 6} \right)^{  L-1 }  } \}\over m } \nonumber \\
	 \label{eq:Proof_MainTheorem_error_bound}
	 & \le C_{11}{ (s \log(N/s) )^{7/2}(\log m)^{12} \cdot m^{ {1 \over 2}\left( {5 \over 6} \right)^{  \lfloor{k/25}\rfloor -1 }  } \cdot m^{ {1 \over 2}\left( {5 \over 6} \right)^{  L-1 }  } \over m }.
	\end{align}
	
	From the definition of $L$, we have 
	$m^{{1 \over 40} ({5 \over 6})^{L+1}} \le 24$ or $ {1 \over 40}  ({5 \over 6})^{L+1} \log m \le \log 24$. This implies that ${1 \over 2}\left({5 \over 6} \right)^{L-1} \lesssim {1 \over \log m}$, i.e., there exists a universal constant $\hat{c} > 0$ such that ${1 \over 2}\left({5 \over 6} \right)^{L-1} \le {\hat{c}  \over \log m}$.
	Since $m^{{1 \over 2}\left({5 \over 6} \right)^{L-1}} \le m^{{\hat{c}  \over \log m}} = \exp \left ({\hat{c}  \over \log m} \cdot \log m \right) = \exp(\hat{c} )$, which is another constant, so applying this fact to \eqref{eq:Proof_MainTheorem_error_bound} implies that
	\[
	C_{11}{ (s \log(N/s) )^{7/2}(\log m)^{12}  \over m^{ \left(1 -{1 \over 2}\left( {5 \over 6} \right)^{ \min\{ \lfloor{k/25}\rfloor, L \}-1 }  \right) } } \le C_{12}{ (s \log(N/s) )^{7/2}(\log m)^{12}  \over m^{ \left(1 -{1 \over 2}\left( {5 \over 6} \right)^{  \lfloor{k/25}\rfloor -1 }  \right) } } 
	\] 
	for some universal constant $C_{12} := C_{11} \cdot \exp(\hat{c}) > 0$. Hence, we establish that
	\[
	\|x_k - x\|_2 \le  C_{12}{ (s \log(N/s) )^{7/2}(\log m)^{12}  \over m^{ \left(1 -{1 \over 2}\left( {5 \over 6} \right)^{  \lfloor{k/25}\rfloor -1 }  \right) } }.
	\]
	
	As for the probability of success, we apply the union bound over all the levels $K_1, K_2, \dots, K_L$ where each holds with probability at least
	\begin{align*}
	& 1 -  \left(3 \bar{c}  \left({s \over N} \right)^{5 s} \cdot {1 \over m^5}  -  {c_9 \over m^4} \right),
	\end{align*} from Theorem \ref{thm: Contraction_with_additivie_error}. Thus, the success probability should be at least 
	\[
	 1 -  L \left(3 \bar{c}  \left({s \over N} \right)^{5 s} \cdot {1 \over m^5}  -  {c_9 \over m^4} \right) \ge 1 - O \left( {1 \over m^3} \right),
	\] where the last inequality is by observing $L = O(\log \log m)$, which can obtained from $m^{{1 \over 40} ({5 \over 6})^L} > 24$. This proves the theorem. 
	
\end{proof}

\section{Proofs for Section \ref{section:Main technical propositions}}
In this section, we prove Lemma \ref{lem:Bound_for_(I)}, \ref{lem:Bound_for_(II)}, and \ref{lem:Bound_for_(III)}. The ideas for their proofs are similar and we start with the proof of Lemma \ref{lem:Bound_for_(II)} because we believe it is simpler than the other two. 

\subsection{Proof of Lemma \ref{lem:Bound_for_(II)}}
\label{subsection:Proof of Lemma II}
For $y \in K^{(i)}$, choose $\hat{y}$ in $\mathcal{N}_{K^{(i)}}$ with $\|y - \hat{y} \| \le  \delta_{i}$. Note that since $y \in K^{(i)}$, $\|x - y\| \le r_i$.

\paragraph{Step 1. Approximate the term $(II)$ by $\epsilon$-net $\mathcal{N}_{K^{(i)}}$} 
Continuing from $(II)$, we obtain 
\begin{align*}
& \left | {1 \over m} \sum_{i=1}^m \tau(\text{sign} (\inner{a_i,x}) - \text{sign}(\inner{a_i, y}) \inner{a_i,{x +y \over \|x +y \|}  } \right| \\
&  \stackrel{(i)} \le \left | {1 \over m} \sum_{i=1}^m \tau(\text{sign} (\inner{a_i, x}) - \text{sign}(\inner{a_i, \hat{y}})) \inner{a_i,{x +y \over \|x +y \|}  } \right| \\
& \qquad  \medmath{ + \left| {1 \over m} \sum\limits_{i=1}^m \left [\tau ( \text{sign}(\inner{a_i,x} ) - \text{sign}(\inner{a_i,y}) \inner{a_i,{x +y \over \|x +y \|}  } \right ] - {1 \over m} \sum\limits_{i=1}^m \left [\tau ( \text{sign}(\inner{a_i, x } ) - \text{sign}(\inner{a_i,\hat{y}})) \inner{a_i,{x +y \over \|x +y \|}  } \right ] \right | } \\
& \stackrel{(ii)}  \le \left | {1 \over m} \sum_{i=1}^m \tau(\text{sign} (\inner{a_i, x}) - \text{sign}(\inner{a_i, \hat{y}})) \inner{a_i,{x +y \over \|x +y \|}  } \right| \\
&\qquad +  {1 \over m} \sum\limits_{i=1}^m  \tau | \text{sign}(\inner{a_i,y} ) - \text{sign}(\inner{a_i, \hat{y} } ) | \left |\inner{a_i,{x +y \over \|x +y \|}  } \right|  \\
&\stackrel{(iii)}   \le \left | {1 \over m} \sum_{i=1}^m \tau(\text{sign} (\inner{a_i, x}) -  \text{sign}(\inner{a_i, \hat{y}}) \inner{a_i,{x +y \over \|x +y \|}  } \right| \\
& \qquad  +  C(N,s,m)  \cdot {\tau \over m} \sum\limits_{i=1}^m  | \text{sign}(\inner{a_i,y} ) - \text{sign}(\inner{a_i, \hat{y} } ) | \\
& \stackrel{(iv)} \le \left | {1 \over m} \sum_{i=1}^m \tau(\text{sign} (\inner{a_i, x}) - \text{sign}(\inner{a_i, \hat{y}})) \inner{a_i,{x +y \over \|x +y \|}  } \right|   +  2C_{10} \tau \cdot \log m \cdot C(N,s,m) \delta_{i},
\end{align*}
as long as we have Lemma \ref{lem: Unified uniform bound} (uniform bound lemma) and Corollary \ref{cor:Local binary embedding} (local binary embedding). Here, in (i) and (ii) follow by the triangle inequality. To have (iii), we used the uniform bound for $ \left| \inner{a_i,{x +y \over \|x +y \|}  } \right| $ in Lemma \ref{lem: Unified uniform bound}. The inequality (iv) is due to Corollary \ref{cor:Local binary embedding}.

Now consider the first term in the right side of (iv):
\begin{align*}
&   \left| {1 \over m } \sum_{i=1}^m \tau(\text{sign} (\inner{a_i, x}) - \text{sign}(\inner{a_i, \hat{y}})) \inner{a_i,{x +y \over \|x +y \|}  } \right| \\ 
& \stackrel{(v)} \le   \left| {1 \over m} \sum_{i=1}^m \tau(\text{sign} (\inner{a_i, x}) - \text{sign}(\inner{a_i, \hat{y}})) \inner{a_i,{ x + \hat{y} \over \|x + y\|}  } \right| \\
& \qquad +     \left| {1 \over m} \sum_{i=1}^m \tau(\text{sign} (\inner{a_i, x}) - \text{sign}(\inner{a_i, \hat{y}})) \left(\inner{a_i,{x +y \over \|x +y \|}  } - \inner{a_i,{ x + \hat{y} \over \|x + y \|}  } \right) \right|\\ 
& \stackrel{(vi)} \le   \left| {1 \over m} \sum_{i=1}^m \tau(\text{sign} (\inner{a_i, x}) - \text{sign}(\inner{a_i, \hat{y}})) \inner{a_i,{ x + \hat{y} \over \|x +y \|}  } \right| \\
& \qquad +      {1 \over m} \sum_{i=1}^m \tau\left| (\text{sign} (\inner{a_i, x}) - \text{sign}(\inner{a_i, \hat{y}})) \right| \left| \left(\inner{a_i,{x +y \over \|x +y \|}  } - \inner{a_i,{ x + \hat{y} \over \|x + y \|}  } \right) \right|\\ 
& \stackrel{(vii)}\le   \left| {1 \over m} \sum_{i=1}^m \tau(\text{sign} (\inner{a_i, x}) - \text{sign}(\inner{a_i, \hat{y}})) \inner{a_i,{ x + \hat{y} \over \|x +y \|}  } \right| \\
& \qquad +   {2  C(N,s,m) \delta_i \over \|x +y \| } \cdot {1 \over m} \sum_{i=1}^m \tau\left| (\text{sign} (\inner{a_i, x}) - \text{sign}(\inner{a_i, \hat{y}})) \right|  \\ 
& \stackrel{(viii)} \le   \left| {1 \over m} \sum_{i=1}^m \tau(\text{sign} (\inner{a_i, x}) - \text{sign}(\inner{a_i, \hat{y}})) \inner{a_i,{ x+ \hat{y} \over \|x +y \|}  } \right|  +   {2 C_{10} \tau { r_i \log m \cdot C(N,s,m)   \delta_{i} \over \|x + y\|_2}} \\
& \stackrel{(ix)} \le   \left| {1 \over m} \sum_{i=1}^m \tau(\text{sign} (\inner{a_i, x}) - \text{sign}(\inner{a_i, \hat{y}})) \inner{a_i,{ x + \hat{y} \over \|x +y \|}  } \right|  +   {4 C_{10} \tau \cdot { r_i \log m \cdot C(N,s,m)   \delta_{i} }}  \\
&  =   \left| {1 \over m} \sum_{i=1}^m \tau(\text{sign} (\inner{a_i, x}) - \text{sign}(\inner{a_i, \hat{y}})) \inner{a_i,{ x + \hat{y} \over \|x + \hat{y} \|}  } \right| \cdot { \| x + \hat{y} \| \over \| x + y \|}  +   {4 C_{10} \tau \cdot { r_i \log m \cdot C(N,s,m)   \delta_{i} }}.
\end{align*}

In the chain of inequalities above, (v) and (vi) follow from the triangle inequality. We applied Lemma \ref{lem: Unified uniform bound} to have (vii). To get (viii), first note that $\|\hat{y} - x\| \le r_i$ since $\hat{y} \in K^{(i)}$ and apply Corollary \ref{cor:Local binary embedding}. (ix) is easily follows from $\|x + y\| \ge 1$ since $y \in K^{(i)}$ implying $\|x- y\| \le r_1 \ll 1$. 

Hence, we obtain
\begin{align*}
    & \left | {1 \over m} \sum_{i=1}^m \tau(\text{sign} (\inner{a_i,x}) - \text{sign}(\inner{a_i, y}) \inner{a_i,{x +y \over \|x +y \|}  } \right| \\
    & \le \left| {1 \over m} \sum_{i=1}^m \tau(\text{sign} (\inner{a_i, x}) - \text{sign}(\inner{a_i, \hat{y}})) \inner{a_i,{ x + \hat{y} \over \|x + \hat{y} \|}  } \right| \cdot { \| x + \hat{y} \| \over \| x + y \|} \\
    & \qquad + 2C_{10} \tau \cdot \log m \cdot C(N,s,m) \delta_{i} +   {4 C_{10} \tau \cdot { r_i \log m \cdot C(N,s,m)   \delta_{i} }} \\
    & \le \left| {1 \over m} \sum_{i=1}^m \tau(\text{sign} (\inner{a_i, x}) - \text{sign}(\inner{a_i, \hat{y}})) \inner{a_i,{ x + \hat{y} \over \|x + \hat{y} \|}  } \right| \cdot { \| x + \hat{y} \| \over \| x + y \|} + {2\tau r_{i+1} \over 600} + {4\tau r_{i+1} \over 600} \\
    & \le \left| {1 \over m} \sum_{i=1}^m \tau(\text{sign} (\inner{a_i, x}) - \text{sign}(\inner{a_i, \hat{y}})) \inner{a_i,{ x + \hat{y} \over \|x + \hat{y} \|}  } \right| \cdot { \| x + \hat{y} \| \over \| x + y \|} + {8 r_{i+1} \over 600},
\end{align*} where we used Proposition \ref{prop: Properties_of_sequences2} in the second inequality.

\paragraph{Step 2. Bounding the mean and variance of truncated terms} 
To keep notation light, let us define a random variable ${\bf{X}_2}(a_i,x,y) $ by
\begin{align}
\label{def:X2}
{\bf{X}_2}(a_i,x,y) =  (\text{sign} (\inner{a_i, x}) - \text{sign}(\inner{a_i, y}) \inner{a_i,{ x + y \over \|x + y \|}  }.
\end{align}

First note that ${\bf{X}_2}(a_i,x,y) $ is mean-zero which can be verified by direct expansion of the right hand side of \eqref{def:X2} and Proposition \ref{prop:one-bit correlation with Gaussian}. This implies
\begin{align*}
0 &= \mathbb{E} \left[{\bf{X}_2}(a_i, x, \hat{y})  \right] \\
 &= \mathbb{E} \left[{\bf{X}_2}(a_i, x, \hat{y}) \ii_{E_i} \right]  	+ \mathbb{E} \left[{\bf{X}_2}(a_i, x, \hat{y}) \ii_{E^c_i} \right]. 
\end{align*}

Hence, we obtain
\begin{align*}
 &\left|	\mathbb{E} \left[ {\bf{X}_2}(a_i, x, \hat{y}) \ii_{E_i} \right] \right| \\
 &\le  \mathbb{E} \left[\left|	 {\bf{X}_2}(a_i, x, \hat{y}) \right| \ii_{E^c_i} \right]   \\
 &\le  \mathbb{E} \left[\left|	 {\bf{X}_2}(a_i, x, \hat{y}) \right|^2  \right]^{1/2}  \cdot [\mathbb{E} \ii_{E^c_i}]^{1/2} \\
 &\le  \mathbb{E} \left[\left|	{\bf{X}_2}(a_i, x, \hat{y}) \right|^2  \right]^{1/2}  \cdot \mathbb{P}(E^c_i)^{1/2} \\
 &= \mathbb{E} \left[  (\text{sign} (\inner{a_i, x}) - \text{sign}(\inner{a_i, \hat{y}}))^2 \left|	\inner{a_i,{ x+ \hat{y} \over \|x + \hat{y} \|}  } \right|^2  \right]^{1/2}  \cdot \mathbb{P}(E^c_i)^{1/2} \\
  &\le \mathbb{E} \left[  (\text{sign} (\inner{a_i, x}) - \text{sign}(\inner{a_i, \hat{y}}))^4 \right]^{1/4} \mathbb{E} \left[	\inner{a_i,{ x + \hat{y} \over \|x + \hat{y} \|}  }^4 \right]^{1/4}  \cdot \mathbb{P}(E^c_i)^{1/2} \\
 &\le 2 \cdot 3^{1/4} \cdot {\sqrt{2} \over m^2}\\
 &\le 4 \cdot {\sqrt{2} \over m^2}.
\end{align*}
Here, the second and fourth inequalities are by the Cauchy-Schwartz inequality. The second last line is due to the facts that $ \left| \text{sign} (\inner{a_i, x}) -\text{sign}(\inner{a_i, \hat{y}}) \right| \le 2$, the fourth moment of the standard Gaussian random variable is $3$, and  $\mathbb{P}(E_i) \ge 1 - {2 \over m^5}$.

By the construction of the event $E_i$ and Lemma \ref{lem: Unified uniform bound}, ${\bf{X}_2}(a_i, x, \hat{y}) \ii_{E_i}$ is bounded by
\[
\left| {\bf{X}_2}(a, x, \hat{y}) \ii_{E_i} \right| \le 2 C(N,s,m)
\] and its second moment is bounded (so is its variance) by 
\begin{align*}
\label{eq:Second_Moment 1}
& \mathop{\mathbb{E}}_{a_i}  \left[{\bf{X}_2}(a, x, \hat{y})^2 \ii_{E_i} \right] \\
&= \mathop{\mathbb{E}}_{a_i}  \left[ (\text{sign} (\inner{a_i, x}) - \text{sign}(\inner{a_i, \hat{y}}))^2 \inner{a_i,{ x + \hat{y} \over \|x + \hat{y} \|}  }^2 \ii_{E_i} \right] \\
& \quad  \le C(N,s,m)^2  \mathop{\mathbb{E}}_{a_i}  \left[ (\text{sign} (\inner{a_i, x}) - \text{sign}(\inner{a_i, \hat{y}}))^2 \ii_{E_i}   \right] \\
& \quad  \le C(N,s,m)^2  \mathop{\mathbb{E}}_{a_i}  \left[ (\text{sign} (\inner{a_i, x}) - \text{sign}(\inner{a_i, \hat{y}}))^2  \right] \\
& \quad  \le 4 C(N,s,m)^2   \cdot d_g(x, \hat{y}) \\
& \quad \le 4 \pi C(N,s,m)^2 \|x - \hat{y}\|_2,
\end{align*} where $d_g(x, \hat{y})$ is the normalized geodesic distance between $x$ and $\hat{y}$. Here, we applied Lemma \ref{lem: Unified uniform bound} to get the first inequality and used $\mathop{\mathbb{E}}_{a_i}  \left[ (\text{sign} (\inner{a_i, x}) - \text{sign}(\inner{a_i, \hat{y}}))^2  \right] = 4 \mathop{\mathbb{P}}_{a_i}(\text{sign} (\inner{a_i, x}) \neq \text{sign}(\inner{a_i, \hat{y}})) = 4 d_g(x, \hat{y})$, which is due to the rotation invariance of the standard Gaussian vector. 

\paragraph{Step 3. Bounding the sum of truncated terms by the Bernstein's inequality}
Next, we apply the Bernstein's inequality for mean-zero bounded random variables to have
\[
\mathbb{P} \left( \left| {1 \over m} \sum_{i=1}^m \tau \left({\bf{X}_2}(a_i, x, \hat{y}) \ii_{E_i} - 	\mathbb{E} \left[ {\bf{X}_2}(a, x, \hat{y}) \ii_{E_i} \right] \right)  \right| \ge t \right) \le 2 \exp \left( - {m^2t^2/2 \over \sigma^2 + Kmt/3} \right)
\] where $\sigma^2$ is the sum of the variances and $K$ is the bound of the random variables $\tau{\bf{X}_2}(a_i, x, \hat{y}) \ii_{E_i}$. 
Because $\sigma^2 \le \tau^2 \cdot m \mathop{\mathbb{E}}_{a_i}  \left[{\bf{X}_2}(a, x, \hat{y})^2 \ii_{E_i} \right]$ and $K \le 2\tau C(N,s,m)$ from Step 2, we obtain
\begin{align*}
&\mathbb{P} \left( \left|{1 \over m} \sum_{i=1}^m \tau \left({\bf{X}_2}(a_i, x, \hat{y}) \ii_{E_i} - 	\mathbb{E} \left[ {\bf{X}_2}(a, x, \hat{y}) \ii_{E_i} \right] \right)  \right| \ge t \right) \\
&\quad \le 2 \exp \left( - {mt^2/2 \over \tau^2 \mathop{\mathbb{E}}_{a_i}  \left[{\bf{X}_2}(a, x, \hat{y})^2 \ii_{E_i} \right] +  2\tau C(N,s,m) t/3} \right) \\
& \quad \le  2 \exp \left( - {m t^2/2 \over 4 \pi \tau^2 C(N,s,m)^2 \cdot  \|x - \hat{y}\|_2  + 2\tau C(N,s,m) t/3} \right)\\
& \quad \le  2 \exp \left( - {m t^2/2 \over 4 \pi \tau^2 C(N,s,m)^2 \cdot  \|x - \hat{y}\|_2  + 4 \pi \tau^2 C(N,s,m)^2 t/3} \right)
\end{align*} where we used the upper estimates of the first and second moments of ${\bf{X}_2}(a, x, \hat{y})\ii_{E_i} $ in Step 2.

Choose $t= {1 \over m^{1/2 - \xi_i}} \|x - \hat{y}\|_2$ in which $\xi_i \in [0, 1/2]$ will be determined later. Then, we have
\begin{align*}
&  {m t^2/2  \over 4 \pi \tau^2 C(N,s,m)^2 \cdot  \|x- \hat{y}\|_2 + 4 \pi  \tau^2  C(N,s,m)^2 t/3} \\
& \quad  = { m^{2 \xi_i} \|x- \hat{y}\|_2^{2} \over 8 \pi  \tau^2  C(N,s,m)^2  \cdot [  \|x - \hat{y}\|_2  +  {1 \over m^{1/2 - \xi_i}} \|x - \hat{y}\|_2/3 ] } \\
& \quad  \ge { m^{2 \xi_i} \|x - \hat{y}\|_2^{2} \over 16 \pi  \tau^2  C(N,s,m)^2  \cdot [ \max \{ \|x - \hat{y}\|_2 ,  {1 \over m^{1/2 - \xi_i}} \|x - \hat{y}\|_2/3  \} ] } \\
& \quad  \ge {1 \over 16 \pi  \tau^2 C(N,s,m)^2  } \cdot   \min \{m^{2\xi_i } \|x - \hat{y}\|_2 ,  3  m^{1/2 + \xi_i} \|x - \hat{y}\|_2\}.
\end{align*}

Using the fact that $r_{i+1}  \le  \|x - \hat{y}\| \le r_i$ since $\hat{y} \in \mathcal{N}_{K^{(i)}} \subset K^{(i)}$, above further reduces to
\begin{align*}
& {1 \over 16 \pi  \tau^2  C(N,s,m)^2  } \cdot   \min \{m^{2\xi_i } \|x - \hat{y}\|_2 ,  3  m^{1/2 + \xi_i} \|x - \hat{y}\|_2 \}  \\
&\quad  \ge {1 \over 16 \pi  \tau^2  C(N,s,m)^2  } \|x - \hat{y}\|_2 \cdot   \min \{m^{2\xi_i }  ,  3  m^{1/2 + \xi_i} \} \\
&\quad  \ge {1 \over 16 \pi  \tau^2 C(N,s,m)^2  } r_{i+1}  \cdot   \min \{m^{2\xi_i } ,  3  m^{1/2 + \xi_i} \} \\
&\quad  \ge {1 \over 80 C(N,s,m)^2  } r_{i+1} m^{2\xi_i },
\end{align*}
where the last inequality is from the fact that $0 \le \xi_i \le 1/2$.
%Now choose  $\xi = 2/5$.

Thus, with a probability at least $1 - 2 \exp \left(- {r_{i+1} m^{2\xi_i }   \over 80 C(N,s,m)^2  }  \right)$, we have

\begin{align}
\label{bound:Deviation bound 1 for (I)}
 \left| {1 \over m} \sum_{i=1}^m \tau \left({\bf{X}_2}(a_i, x, \hat{y}) \ii_{E_i} - 	\mathbb{E} \left[ {\bf{X}_2}(a, x, \hat{y}) \ii_{E_i} \right] \right)  \right| \le {1 \over m^{1/2 - \xi_i }} \|x - \hat{y}\|_2.
\end{align}

By the union bound over all  $ \hat{y}$ in $\mathcal{N}_{K^{(i)}}$, 
above bound holds uniformly for all $\hat{y}  \in \mathcal{N}_{K^{(i)}}$ with probability at least 
\begin{align*}
&1 - 2 | \mathcal{N}_{K^{(i)}} |  \exp \left(- {r_{i+1} m^{2\xi_i }   \over 80 C(N,s,m)^2  }  \right ) \\ 
& \ge 1 -2 \exp \left(C^{''}  (600C_{10} )^{2 \left(1 -\left({5 \over 6} \right)^{i} \right)}m^{{1 \over 3}({5 \over 6})^{i-1}} (\log m)^{2 \left( 1- {7 \over 3}\left({5 \over 6} \right)^{i-1} \right)} C(N,s,m)^{4  - {10 \over 3}\left({5 \over 6} \right)^{i-1}}  \right) \\
& \qquad \qquad  \times  \exp \left(- (600C_{10})^{3 \left(1 -\left({5 \over 6} \right)^{i} \right)} {m^{2\xi_i - \left(1 -{1 \over 2}\left({5 \over 6} \right)^{i} \right)}  (\log m)^{7 \left(1 -  \left({5 \over 6} \right)^{i} \right)} C(N,s,m)^{6 - 5\left({5 \over 6} \right)^i }  \over 80 C(N,s,m)^2  }  \right)
\end{align*} where the bound for $| \mathcal{N}_{K^{(i)}} |$ is from Subsection \ref{section:Epsilon_net} and $r_{i+1}$ is from Definition \ref{def:Sequences_for_nets}.

Note that $\xi_i$ should satisfy $2\xi_i - \left(1 - {1 \over 2}({5 \over 6})^i \right) \ge {3 \over 8}({5 \over 6})^{i-1}$ for this to hold with a high probability. Now choose $\xi := {1 \over 2} \left(1 - {1 \over 20} ({5 \over 6})^i \right)$ which guarantees that
\begin{align*}
&{ C^{''} (600C_{10})^{2 \left(1 -\left({5 \over 6} \right)^{i-1} \right)}m^{{1 \over 3}({5 \over 6})^{i-1}} (\log m)^{2 \left( 1- {7 \over 3}\left({5 \over 6} \right)^{i-1} \right)} C(N,s,m)^{4  - {10 \over 3}\left({5 \over 6} \right)^{i-1}}} \\
&\qquad -  (600C_{10})^{3 \left(1 -\left({5 \over 6} \right)^{i} \right)}m^{{3 \over 8}({5 \over 6})^{i-1} }  (\log m)^{7 \left(1 -  \left({5 \over 6} \right)^{i} \right)} C(N,s,m)^{4 - 5 \left({5 \over 6} \right)^i } \\
& \le  C^{(3)} m^{{1 \over 3}({5 \over 6})^{i-1}} \left[  (\log m)^{2 \left( 1- {7 \over 3}\left({5 \over 6} \right)^{i} \right)} C(N,s,m)^{4  - 4 \left({5 \over 6} \right)^{i}} -    m^{{1 \over 24}({5 \over 6})^{i-1}} (\log m)^{7 \left(1 -  \left({5 \over 6} \right)^{i} \right)} C(N,s,m)^{4 - 5 \left({5 \over 6} \right)^i } \right] \\
& = C^{(3)} m^{{1 \over 3}({5 \over 6})^{i-1}}C(N,s,m)^{4  - 4 \left({5 \over 6} \right)^{i}} \left[  (\log m)^{2 \left( 1- {7 \over 3}\left({5 \over 6} \right)^{i} \right)}  -  (\log m)^{7 \left(1 -  \left({5 \over 6} \right)^{i} \right)}  m^{{1 \over 20}({5 \over 6})^{i}}  C(N,s,m)^{- \left({5 \over 6} \right)^i } \right] \\
& \le C^{(4)} m^{{1 \over 3}({5 \over 6})^{i-1}}C(N,s,m)^{4  - 4 \left({5 \over 6} \right)^{i}} \left[  (\log m)^{2 \left( 1- {7 \over 3}\left({5 \over 6} \right)^{i} \right)}  -  (\log m)^{7 \left(1 -  \left({5 \over 6} \right)^{i} \right)}   \right] \\
& \le -  {C^{(4)} \over 2}m^{{1 \over 3}({5 \over 6})^{i-1} }  (\log m)^{7 \over 6} \\
& \le -  {C^{(4)} \over 2}  (\log m)^{7 \over 6}
\end{align*} for some small fixed constants $C^{(3)}, C^{(4)} > 0$.
Here the second inequality follows from $m^{1 \over 20} \ge  C(N,s,m)$ if $m > c s^{10} \log^{10} (N/s)$ for a sufficiently large constant $c > 0$. 

Hence, we have 
\begin{align*}
&1 - 2 |  \mathcal{N}_{K^{(i)}} |  \exp \left(- {r_{i+1} m^{2\xi_i }   \over 80 C(N,s,m)^2 }  \right ) \ge 1 -2 \exp \left( - {C^{(4)}  \over 2}  (\log m)^{7 \over 6}  \right) \ge 1 - {c_4 \over m^4}
\end{align*} for some absolute constant $c_4 > 0$ and for all $m > c s^{10} \log^{10} (N/s)$  where $c$ is a sufficiently large constant.

\paragraph{Step 4. Establishing the bound for $(II)$}

Continuing from Step 1, recall that $(II)$ is now bounded by
\begin{align*}
& \left | {1 \over m} \sum_{i=1}^m \tau(\text{sign} (\inner{a_i,x}) - \text{sign}(\inner{a_i, y}) \inner{a_i,{x +y \over \|x +y \|}  } \right| \\
& \quad \le \left| {1 \over m} \sum_{i=1}^m \tau(\text{sign} (\inner{a_i, x}) - \text{sign}(\inner{a_i, \hat{y}})) \inner{a_i,{ x + \hat{y} \over \| x  + \hat{y} \|}  } \right| \cdot { \| x + \hat{y} \| \over \| x + y \|}  +   {8 r_{i+1} \over 600}\\
& \quad \le 2 \left| {1 \over m} \sum_{i=1}^m \tau{\bf{X}_2}(a_i, x, \hat{y}) \right|  +    {8 r_{i+1} \over 600}.
\end{align*}

Also by the definition of the event $E_u$, whenever the event $E_u$ occurs, we have
\[
{1 \over m} \sum_{i=1}^m \tau{\bf{X}_2}(a_i, x, \hat{y}) = {1 \over m} \sum_{i=1}^m \tau{\bf{X}_2}(a_i, x, \hat{y}) \ii_{E_i}.
\]

Hence, by combining all the results together, for all $y$ with $r_{i+1} \le \|x-y\| \le r_i$, we have
\begin{align*}
& \left | {1 \over m} \sum_{i=1}^m \tau(\text{sign} (\inner{a_i,x}) - \text{sign}(\inner{a_i, y}) \inner{a_i,{x +y \over \|x +y \|}  } \right| \\
&  \quad \stackrel{(a)}  \le 2 \left| {1 \over m} \sum_{i=1}^m \tau{\bf{X}_2}(a_i, x, \hat{y})  \ii_{E_i} \right|  +    {8 r_{i+1} \over 600}\\
& \quad \stackrel{(b)}  \le 2 \tau \left| {1 \over m} \sum_{i=1}^m {\bf{X}_2}(a_i, x, \hat{y}) \ii_{E_u} -  \mathbb{E} \left[ {\bf{X}_2}(a, x, \hat{y}) \ii_{E_u} \right]  \right| +  2\tau  \left| \mathbb{E} \left[ {\bf{X}_2}(a, x, \hat{y}) \ii_{E_u} \right] \right| +   {8 r_{i+1} \over 600}\\
& \quad \stackrel{(c)}  \le   {2 \over m^{{1 \over 40} ({5 \over 6})^i }} \|x - \hat{y}\|_2 + {8 \sqrt{2} \tau  \over m^2} +  {8 r_{i+1} \over 600}\\
& \quad \stackrel{(d)}  \le  {2 \over m^{{1 \over 40} ({5 \over 6})^i }}(\|x - y\|_2 +  \delta_i) +  {8\sqrt{2} \tau  \over m^2}  +    {8 r_{i+1} \over 600}\\
& \quad \stackrel{(e)} \le    {2 \over m^{{1 \over 40} ({5 \over 6})^i }}\|x - y\|_2 + 2\delta_i  + {8 \sqrt{2} \tau  \over m^2}  +   {8 r_{i+1} \over 600}\\
& \quad \stackrel{(f)}   \le    {2 \over m^{{1 \over 40} ({5 \over 6})^i }}\|x - y\|_2 + {2 \over 600}r_{i+1}  + {1 \over 600}r_{i+1}  +   {8 r_{i+1} \over 600}\\
& \quad \le    {2 \over m^{{1 \over 40} ({5 \over 6})^i }}\|x - y\|_2 +  {12 r_{i+1} \over 600}
\end{align*} for all $m > c s^{10} \log^{10} (N/s)$.
Here, the first inequality follows from the assumption that $E_u$ occurs. We applied the triangle inequality to have (b), and (d). The inequality (c) results from \eqref{bound:Deviation bound 1 for (I)}, $\xi = {1 \over 2} \left(1 - {1 \over 20} ({5 \over 6})^i \right)$, and the bound for $ \left| \mathbb{E} \left[ {\bf{X}_2}(a, x, \hat{y}) \ii_{E_u} \right]\right|$ in Step 2. The inequality (e) follows from $m^{{1 \over 40} ({5 \over 6})^i} > m^{{1 \over 40} ({5 \over 6})^L} > 24 $. We applied Proposition \ref{prop: Properties_of_sequences2} to get (f).

Combining the results we have so far, $(II)$ is bounded by
\begin{align*}
(II) \le {2 \over m^{{1 \over 40} ({5 \over 6})^i }}\|x - y\|_2  + {r_{i+1} \over 50}.
\end{align*}
As for the success probability for this bound to hold, we need Lemma \ref{lem: Unified uniform bound}, Corollary \ref{cor:Local binary embedding}, the event $E_u$ occurred, the concentration bound in Step 3, and the error bound for the first iteration of NBIHT \eqref{eq:bound_for_first_NBHIT_iteration}. Hence, applying the union bound gives us
\begin{align*}
& 1 - {2 \over m^4} - {c_2 \over m^5} - {2 \over m^4} - {c_4 \over m^4} - \bar{c} \left({s \over N} \right)^{5 s} \cdot {1 \over m^5}  \ge 1 -   \bar{c}  \left({s \over N} \right)^{5 s} \cdot {1 \over m^5} - {c_6 \over m^4}.
\end{align*}

\subsection{Proof of Lemma \ref{lem:Bound_for_(I)}}
\label{subsection:Proof of Lemma I}

This subsection is devoted to derive the bound for the term $(I)$. The idea is quite similar to the proof of Lemma \ref{lem:Bound_for_(II)}, so we omitted some parts of the proof to avoid repetitions. 

As before let $y \in K^{(i)}$, so $r_{i+1} \le \|x -y\| \le r_i$. We proceed with the similar arguments used for the bound of $(II)$ in subsection \ref{subsection:Proof of Lemma II}.

\paragraph{Step 1. Approximate $(I)$ by $\epsilon$-net $\mathcal{N}_{K^{(i)}}$} 
We begin with approximating $(I)$ with the $\epsilon$-net $\mathcal{N}_{K^{(i)}}$. The following inequality holds
\begin{align}
\nonumber
&\left | {1 \over m} \sum_{i=1}^m \tau (\text{sign} (\inner{a_i,x}) - \text{sign}(\inner{a_i, y}) \inner{a_i, {x -y \over \|x -y \| }}  -  \|x- y\| \right| \\
\label{term:term1_in_Bound_I}
&\le \left | {1 \over m} \sum_{i=1}^m \tau (\text{sign} (\inner{a_i,x }) - \text{sign}(\inner{a_i, \hat{y}}) \inner{a_i, {x -\hat{y} \over \|x -\hat{y} \| }}  -  \|x -\hat{y}\| \right| \\
\label{term:term2_in_Bound_I}
& + \medmath{ \left | {1 \over m} \sum_{i=1}^m \tau (\text{sign} (\inner{a_i,x }) - \text{sign}(\inner{a_i, \hat{y}}) \inner{a_i, {x -\hat{y} \over \|x -\hat{y} \| }} - {1 \over m} \sum_{i=1}^m \tau (\text{sign} (\inner{a_i,x}) - \text{sign}(\inner{a_i, y}) \inner{a_i, {x -y \over \|x -y \| }}  \right| } \\
\nonumber
&  + \left| \|x- y\|  -  \|x -\hat{y}\| \right|
\end{align}
provided we have Lemma \ref{lem: Unified uniform bound} and Corollary \ref{cor:Local binary embedding}.
%with probability exceeding $1- {2 \over m^4} - {c_2 \over m^4} \ge 1- {c_3 \over m^4}$.

The third term in the right hand side is bounded by the triangle inequality:
\[
\left| \|x- y\|  -  \|x -\hat{y}\| \right| \le \|x- y - x + \hat{y}\| = \| y  - \hat{y}\| \le  \delta_i. 
\]

\paragraph{Step 2. Bounding the mean and variance of truncated terms}
Define a random variable ${\bf{X}_1}(a_i,x,y) $ as
\[
{\bf{X}_1}(a_i,x,y) =  (\text{sign} (\inner{a_i, x}) - \text{sign}(\inner{a_i, y}) \inner{a_i,{ x - y \over \|x - y \|}  }.
\]

As in the previous subsection, we have
\[
\tau \mathbb{E} \left[ \text{sign} (\inner{a_i, x}) - \text{sign}(\inner{a_i, y}) \inner{a_i,{ x - y }} \right] = 2 - 2 \inner{x,y} = \|x - y\|_2^2
\] by directly expanding terms in the expectation and applying Proposition \ref{prop:one-bit correlation with Gaussian}. 

Hence, we obtain
\begin{align*}
\|x- \hat{y}\|_2 &= \mathbb{E} \left[\tau {\bf{X}_1}(a_i, x, \hat{y})  \right] \\
&= \mathbb{E} \left[\tau {\bf{X}_2}(a_i, x, \hat{y}) \ii_{E_i} \right]  	+ \mathbb{E} \left[\tau{\bf{X}_2}(a_i, x, \hat{y}) \ii_{E^c_i} \right]. 
\end{align*}
So, the same arguments in the previous section give us the following bound.
\begin{align*}
&\left|	\mathbb{E} \left[\tau {\bf{X}_1}(a_i, x, \hat{y}) \ii_{E_i} \right] - \|x- \hat{y}\|_2 \right| \\
&\le  \mathbb{E} \left[\left|\tau	 {\bf{X}_1}(a_i, x, \hat{y}) \right| \ii_{E^c_i} \right]   \\
&\le  {4c_5 \over m^2}. 
\end{align*}

By the construction of the event $E_i$, ${\bf{X}_1}(a_i, x, \hat{y}) \ii_{E_i}$ is bounded by
\[
\left| {\bf{X}_1}(a, x, \hat{y}) \ii_{E_i} \right| \le 2C(N,s,m)
\] and the bound for the second moment for ${\bf{X}_1}(a_i, x, \hat{y}) \ii_{E_i}$ is given by
\begin{align*}
\label{eq:Second_Moment 2}
& \mathop{\mathbb{E}}_{a_i}  \left[{\bf{X}_1}(a, x, \hat{y})^2 \ii_{E_i} \right]  \le 4 \pi C(N,s,m)^2 \|x - \hat{y}\|_2
\end{align*} by the same argument used for ${\bf{X}_2}(a_i, x, \hat{y}) \ii_{E_i}$ in the previous subsection.

\paragraph{Step 3. Bounding the sum of truncated terms by the Bernstein's inequality}

Step 2 shows that the magnitude and variance of ${\bf{X}_1}(a_i, x, \hat{y}) \ii_{E_i}$ can be bounded exactly by those of ${\bf{X}_2}(a_i, x, \hat{y}) \ii_{E_i}$ in Step 2 in Subsection \ref{subsection:Proof of Lemma II}. Hence, applying the bounded Bernstein inequality and using the $\epsilon$-net covering argument for $K^{(i)}$ as in Subsection \ref{subsection:Proof of Lemma II} yield the following statement:

For all $\hat{y} \in \mathcal{N}_{K^{(i)}}$, we obtain

\[
\left| {1 \over m} \sum_{i=1}^m \tau {\bf{X}_1}(a_i, x, \hat{y}) \ii_{E_i} - 	\mathbb{E} \left[ \tau {\bf{X}_1}(a, x, \hat{y}) \ii_{E_i} \right]   \right| \le {1 \over m^{1/2 - \xi_i }} \|x - \hat{y}\|_2
\] with probability $1 - {c_4 \over m^4}$.

Then, we apply the triangle inequality to above to get
\begin{align*}
\left| {1 \over m} \sum_{i=1}^m \tau {\bf{X}_1}(a_i, x, \hat{y}) \ii_{E_i} - 	\|x - \hat{y}\|_2 	 \right| 
& \le {1 \over m^{1/2 - \xi_i }} \|x - \hat{y}\|_2 + {4 c_5 \over m^2}	\\
& \le {1 \over m^{1/2 - \xi_i }} \|x -y\|_2  +\|y- \hat{y}\| + {r_{i+1} \over 600}	\\
& \le {1 \over m^{1/2 - \xi_i }} \|x -y\|_2  + \delta_i + {r_{i+1} \over 600}	\\
& \le {1 \over m^{{1 \over 40}({5 \over 6})^i}} \|x -y\|_2 + {2r_{i+1} \over 600}	\\
\end{align*}

for all $m > c s^{10} \log^{10} (N/s)$ with a probability at least $1 - {c_4 \over m^4}$.

\paragraph{Step 4. Establishing the bound for $(I)$}

As we derived the bound for for $\tau{\bf{X}_2}(a_i, x, \hat{y})$ in Section \ref{subsection:Proof of Lemma II}, whenever the event $E_u$ occurs, we have
\[
{1 \over m} \sum_{i=1}^m \tau{\bf{X}_1}(a_i, x, \hat{y}) = {1 \over m} \sum_{i=1}^m \tau{\bf{X}_1}(a_i, x, \hat{y}) \ii_{E_i}.
\]
This conditional equality and the bound for the truncated terms in \eqref{term:term1_in_Bound_I} in Step 3 and allow us to control \eqref{term:term1_in_Bound_I}.

On the other hand, the term \eqref{term:term2_in_Bound_I} is bounded as below:
\begin{align*}
&  \left | {1 \over m} \sum_{i=1}^m \tau (\text{sign} (\inner{a_i,x }) - \text{sign}(\inner{a_i, \hat{y}}) \inner{a_i, {x -\hat{y} \over \|x -\hat{y} \| }} - {1 \over m} \sum_{i=1}^m \tau (\text{sign} (\inner{a_i,x}) - \text{sign}(\inner{a_i, y}) \inner{a_i, {x -y \over \|x -y \| }}  \right|  \\
& \stackrel{(a)}  \le \left | {1 \over m} \sum_{i=1}^m \tau (\text{sign} (\inner{a_i,x }) - \text{sign}(\inner{a_i, \hat{y}}) \left(\inner{a_i, {x -\hat{y} \over \|x -\hat{y} \| }} - \inner{a_i, {x -y \over \|x -y \| }} \right) \right| \\
&\qquad + \left | {1 \over m} \sum_{i=1}^m \tau  (\text{sign} (\inner{a_i,x }) - \text{sign}(\inner{a_i, \hat{y} } -\text{sign} (\inner{a_i, x }) + \text{sign}(\inner{a_i, y} ) \inner{a_i, {x -y \over \|x -y \| }} \right| \\
& \stackrel{(b)} \le {1 \over m} \sum_{i=1}^m \tau | \text{sign} (\inner{a_i,x }) - \text{sign}(\inner{a_i, \hat{y}} | \left| \inner{a_i, {x -\hat{y} \over \|x -\hat{y} \| }-  {x -y \over \|x -y \| }} \right| \\
&\qquad +  {1 \over m} \sum_{i=1}^m \tau  \left | \text{sign} (\inner{a_i,\hat{y} })  -\text{sign} (\inner{a_i, y })  \right| \left | \inner{a_i, {x -y \over \|x -y \| }} \right|  \\
&  \stackrel{(c)} \le {1 \over m} \sum_{i=1}^m \tau | \text{sign} (\inner{a_i,x }) - \text{sign}(\inner{a_i, \hat{y}} | \left| \inner{a_i, {x -\hat{y} \over \|x -\hat{y} \| }-  {x -y \over \|x -y \| }} \right| +  2 \tau C_{10} { C(N,s,m) \log m \cdot  \delta_{i} }\\
&  \stackrel{(d)} \le  {1 \over m} \sum_{i=1}^m \tau | \text{sign} (\inner{a_i,x }) - \text{sign}(\inner{a_i, \hat{y}} | \cdot   C(N,s,m)\cdot {2 \delta_{i} \over r_{i+1}}   +   2 \tau C_{10} C(N,s,m) \log m \cdot  \delta_{i} \\
&  \stackrel{(e)}\le   C(N,s,m) \log m \cdot {4 \tau r_i \delta_{i} \over r_{i+1}}   +  2 \tau C_{10} C(N,s,m) \log m \cdot  \delta_{i} \\
&  \stackrel{(f)} \le   {4 \tau  r_{i+1} \over 600} +  2 \tau C_{10} C(N,s,m) \log m \cdot  \delta_{i}\\
&  \stackrel{(g)} \le   {4 \tau  r_{i+1} \over 600} +  {2 \tau  r_{i+1} \over 600}\\
&   \le  {8 \over 600} r_{i+1}.
\end{align*}
Here are the justifications for the above chain of inequalities: (a) and (b) are by the triangle inequality. We applied Lemma \ref{lem: Unified uniform bound} and Corollary \ref{cor:Local binary embedding} to obtain (c). The inequality (d) results from Proposition \ref{prop:difference_of_two_normalized_vectors} and Lemma \ref{lem: Unified uniform bound}. (e) arises from applying Corollary \ref{cor:Local binary embedding}. In (f) and (g), we used Proposition \ref{prop: Properties_of_sequences2}.

Putting all the bounds we have so far together, we establish
\begin{align*}
	(I) & \le {1 \over m^{{1 \over 40}({5 \over 6})^i}} \|x -y\|_2 + {2r_{i+1} \over 600} + \delta_i + {8 \over 600} r_{i+1}\\
	& \le {1 \over m^{{1 \over 40}({5 \over 6})^i}} \|x -y\|_2  + {r_{i+1} \over 50} 
\end{align*}
with probability at least 
\begin{align*}
& 1 - {2 \over m^4} - {c_2 \over m^5} - {2 \over m^4} - {c_4 \over m^4} - \bar{c} \left({s \over N} \right)^{5 s} \cdot {1 \over m^5}  \ge 1 -   \bar{c}  \left({s \over N} \right)^{5 s} \cdot {1 \over m^5} - {c_6 \over m^4}.
\end{align*}

\subsection{Proof of Lemma \ref{lem:Bound_for_(III)}}
As before, let $y \in K^{(i)}$, so $r_{i+1} \le \|x -y\| \le r_i$ and $\hat{y} \in \mathcal{N}_{K^{(i)}}$ with $\|y - \hat{y} \| \le \delta_i$. 

\paragraph{Step 1. Approximate $(III)$ by $\epsilon$-net $\mathcal{N}_{K^{(i)}}$} First, we apply the triangle inequality and Lemma \ref{lem:Orthogonal_decomposition_Gaussian_vector} to have
\begin{align*}
& \left \| {\tau \over m} \sum_{i=1}^m(\text{sign} (\inner{a_i,x}) - \text{sign}(\inner{a_i, y}) b_i(x,y) \right \|_{K_1^\circ}  \\
& \le \left \| {\tau \over m} \sum_{i=1}^m(\text{sign} (\inner{a_i,x }) - \text{sign}(\inner{a_i, y}) b_i(x, \hat{y}) \right \|_{K_t^\circ}  \\
& \qquad + \left \|  {\tau \over m} \sum_{i=1}^m(\text{sign} (\inner{a_i,x}) - \text{sign}(\inner{a_i, y}) ( b_i(x,y) - b_i(x, \hat{y})) \right \|_{K_1^\circ} \\
& \le \left \| {\tau \over m} \sum_{i=1}^m(\text{sign} (\inner{a_i,x }) - \text{sign}(\inner{a_i, y}) b_i(x, \hat{y}) \right \|_{K_1^\circ}  \\
&\qquad  + \left \|  {\tau \over m} \sum_{i=1}^m(\text{sign} (\inner{a_i,x}) - \text{sign}(\inner{a_i, y}) \left( \inner{a_i, {x -y \over \|x -y \| }} {x -y \over \|x -y \| } - \inner{a_i, {x -\hat{y} \over \|x -\hat{y}  \| }} {x -\hat{y}  \over \|x -\hat{y}  \| } \right) \right \|_{K_1^\circ}  \\
&\qquad  + \left \|  {\tau \over m} \sum_{i=1}^m(\text{sign} (\inner{a_i,x}) - \text{sign}(\inner{a_i, y}) \left( \inner{a_i, {x +y \over \|x +y \| }} {x +y \over \|x +y \| } - \inner{a_i, {x +\hat{y} \over \|x + \hat{y}  \| }} {x + \hat{y}  \over \|x +\hat{y}  \| } \right) \right \|_{K_1^\circ} 
\end{align*}

We will bound the second and third term using Lemma \ref{lem:stable_perturbation_lemma2}.

By applying Lemma \ref{lem:stable_perturbation_lemma2}, Corollary \ref{cor:Local binary embedding}, and Proposition \ref{prop: Properties_of_sequences2} to the second term, we have
\begin{align*}
&  \left \|  {\tau \over m} \sum_{i=1}^m(\text{sign} (\inner{a_i,x}) - \text{sign}(\inner{a_i, y}) \left( \inner{a_i, {x -y \over \|x -y \| }} {x -y \over \|x -y \| } - \inner{a_i, {x -\hat{y} \over \|x -\hat{y}  \| }} {x -\hat{y}  \over \|x -\hat{y}  \| } \right) \right \|_{K_1^\circ}  \\
&\le  {\tau \over m} \sum_{i=1}^m \left| \text{sign} (\inner{a_i,x}) - \text{sign}(\inner{a_i, y} \right| \left \|  \left( \inner{a_i, {x -y \over \|x -y \| }} {x -y \over \|x -y \| } - \inner{a_i, {x -\hat{y} \over \|x -\hat{y}  \| }} {x -\hat{y}  \over \|x -\hat{y}  \| } \right) \right \|_{K_1^\circ} \\ 
& \le  r_{i} (\log m) \cdot  {4 \tau C(N,s,m)  \delta_i\over r_{i+1}}\\
& \le {4 \tau \over 600}r_{i+1}.
\end{align*}
Similarly, the third term is bounded by ${4 \tau  \over 600}r_{i+1}$.

It remains to show that the first term is well-controlled.
\begin{align*}
 & \left \| {\tau \over m} \sum_{i=1}^m(\text{sign} (\inner{a_i,x }) - \text{sign}(\inner{a_i, y}) b_i(x, \hat{y}) \right \|_{K_1^\circ}  \\
 & \le \left \| {\tau \over m} \sum_{i=1}^m(\text{sign} (\inner{a_i, x}) - \text{sign}(\inner{a_i, \hat{y} }) b_i(x, \hat{y}) \right \|_{K_1^\circ}  \\
 & \qquad +  \left \| {\tau \over m} \sum_{i=1}^m(\text{sign} (\inner{a_i,x }) - \text{sign}(\inner{a_i, y}) - \text{sign} (\inner{a_i, x}) + \text{sign}(\inner{a_i, \hat{y} }) b_i(x, \hat{y}) \right \|_{K_1^\circ} \\
  & \le \left \| {\tau \over m} \sum_{i=1}^m(\text{sign} (\inner{a_i, x}) - \text{sign}(\inner{a_i, \hat{y} }) b_i(x, \hat{y}) \right \|_{K_1^\circ}   + {\tau \over m} \sum_{i=1}^m \left | \text{sign} (\inner{a_i,y })  - \text{sign} (\inner{a_i, \hat{y}}) \right|  \left \| b_i(x, \hat{y}) \right \|_{K_1^\circ}  \\
& \le \left \| {\tau \over m} \sum_{i=1}^m(\text{sign} (\inner{a_i, x}) - \text{sign}(\inner{a_i, \hat{y} }) b_i(x, \hat{y}) \right \|_{K_1^\circ}   + 2 \tau \delta_i \log m \cdot  \left \| b_i(x, \hat{y}) \right \|_{K_1^\circ}  \\
& \le \left \| {\tau \over m} \sum_{i=1}^m(\text{sign} (\inner{a_i, x}) - \text{sign}(\inner{a_i, \hat{y} }) b_i(x, \hat{y}) \right \|_{K_1^\circ}   +     2 \tau C(N,s,m)\log m \cdot \delta_{i} \\
\numberthis \label{ineq:Lemma23_first_term}
& \le \left \| {\tau \over m} \sum_{i=1}^m(\text{sign} (\inner{a_i, x}) - \text{sign}(\inner{a_i, \hat{y} }) b_i(x, \hat{y}) \right \|_{K_1^\circ}   +   { 2 \tau r_{i+1} \over 600}
\end{align*}
whenever Lemma \ref{lem: Unified uniform bound} and Corollary \ref{cor:Local binary embedding} hold. Here, the second last inequality is from Lemma \ref{lem: Unified uniform bound} and the last is by Proposition \ref{prop: Properties_of_sequences2}. Thus, it boils down to control the first term in the right hand side of \eqref{ineq:Lemma23_first_term}, which is presented in the next step. 

\paragraph{Step 2. Bounding the first term by decoupling} \mbox{}\\
As for the first term $ \left \| {\tau \over m} \sum_{i=1}^m(\text{sign} (\inner{a_i, x}) - \text{sign}(\inner{a_i, \hat{y} }) b_i(x, \hat{y}) \right \|_{K_1^\circ}$, we apply a simple variant of Lemma 8.1 in \cite{plan2016high}. This implies that 
$\inner{a_i, x}, \inner{a_i, \hat{y} }$ and $b_i(x, \hat{y})$ are independent, which consequently shows that $ (\text{sign} (\inner{a_i, x}) - \text{sign}(\inner{a_i, \hat{y} })$ and 
$b_i(x, \hat{y})$ are independent. This allows us to apply the concentration inequality 
conditioned on $ (\inner{a_i, x}) - \text{sign}(\inner{a_i, \hat{y} })$ as we describe below.

Define
\[
h(x, \hat{y}) := {1\over m}  \sum_{i=1}^m (\text{sign} (\inner{a_i, x}) - \text{sign}(\inner{a_i, \hat{y} }) b_i(x, \hat{y}).
\]
This object has a very similar structure as the function $h$ in Section 8.4.3 in \cite{plan2016high} and we are going to follow the arguments appeared in Section 8.4.3 and 9.1 in that paper.  

Conditioned on $\text{sign} (\inner{a_i, x}) - \text{sign}(\inner{a_i, \hat{y} }$, we have

\[
h \sim N(0, \lambda I_{(x, \hat{y})^\perp}) \qquad \text{where} \qquad  \lambda^2 := {1\over m^2}  \sum_{i=1}^m (\text{sign} (\inner{a_i, x}) - \text{sign}(\inner{a_i, \hat{y} })^2. 
\]

By following the argument in Section 8.4.3 of \cite{plan2016high}, conditional on $\text{sign} (\inner{a_i, x}) - \text{sign}(\inner{a_i, \hat{y} }$, we have
\[
\mathbb{E} \|h \|_{K_t^\circ} \le \lambda \cdot w(K_t).
\]
Since
\[
\lambda = {1\over m}  \sqrt{\sum_{i=1}^m (\text{sign} (\inner{a_i, x}) - \text{sign}(\inner{a_i, \hat{y} })^2} = {\sqrt{2} \over \sqrt{m}}  \sqrt{{1\over m}\sum_{i=1}^m \left| \text{sign} (\inner{a_i, x}) - \text{sign}(\inner{a_i, \hat{y} }\right| },
\] we have 
$\lambda \le {2 \over m^{1/2}} \sqrt{C_{10} r_i \log m}$ with probability $1-{c_2 \over m^4}$ by Corollary \ref{cor:Local binary embedding}.

Again, by the same consideration in \cite{plan2016high} based on the Gaussian concentration inequality (Section 8.3 and 9.1 in \cite{plan2016high} to control the term $E_2$) for any $\epsilon \ge 2m^{-1/2}$, we have
\begin{align*}
\left \| {\tau \over m} \sum_{i=1}^m(\text{sign} (\inner{a_i, x}) - \text{sign}(\inner{a_i, \hat{y} }) b_i(x, \hat{y}) \right \|_{K_1^\circ} 
& \le 4 \tau  \left({C_{10} r_i \log m \over m}\right)^{1/2}  \left(  w(K_1) +  \epsilon m^{1/2} \right)\\
& \le 4 \tau  \left({C_{10} r_i \log m \over m}\right)^{1/2}  \left( C(N,s,m)^{1/2} + \epsilon m^{1/2} \right)\\
& \le 4 \tau  \left({C_{10} r_i \log m \over m}\right)^{1/2}  \left( C(N,s,m)^{1/2}  \cdot \epsilon m^{1/2} \right)
\end{align*}
with probability at least $1 - \exp (-c' \epsilon^2 m)$ for some universal constant $c' > 0$  (in the last inequality, we have used the simple fact that $ab \ge a + b$ if $a,b \ge 2$).

From Proposition \ref{prop: Properties_of_sequences}, 
\begin{align*}
\left \| {\tau \over m} \sum_{i=1}^m(\text{sign} (\inner{a_i, x}) - \text{sign}(\inner{a_i, \hat{y} }) b_i(x, \hat{y}) \right \|_{K_1^\circ} 
& \le 4 {r_{i+1} \over 10\sqrt{6}} \cdot {m^{-1/2} \over \delta^{1/2}_{i}} \cdot \epsilon m^{1/2}\\
& \le  {2 r_{i+1} \over 10} \cdot {\epsilon \over \delta^{1/2}_{i}}.
\end{align*}

Taking $\epsilon := {\delta^{1/2}_i  \over \sqrt{60}}$ (note that this choice of $\epsilon \gg 2m^{-1/2}$ from the construction of $\delta_i$) yields
\begin{align}
\label{eq:Bound_(III)_for_one_pair}
\left \| {\tau \over m} \sum_{i=1}^m(\text{sign} (\inner{a_i, x}) - \text{sign}(\inner{a_i, \hat{y} }) b_i(x, \hat{y}) \right \|_{K_1^\circ}
& \le   {2 \tau  r_{i+1} \over 600}
\end{align} with probability at least 
\begin{align*}
& 1 - \exp (-c' \epsilon^2 m) \\
& \quad = 1 -\exp \left(-c'  {\delta_i m^{1 } \over 60}  \right) \\
& \quad  = 1 -\exp \left(-c'  (600C)^{3 \left(1 -\left({5 \over 6} \right)^{i-1} \right)}  { m^{ {1 \over 3}\left({5 \over 6} \right)^{i-1} } (\log m)^{ {14 \over 3}\left(1 -\left({5 \over 6} \right)^{i-1} \right) +{4 \over 3}} C(N,s,m)^{5  - {10 \over 3}\left({5 \over 6} \right)^{i-1}}\over 60}  \right).
\end{align*}
Since we want the bound \eqref{eq:Bound_(III)_for_one_pair} holds for all $\hat{y}$ in $\mathcal{N}_{K^{(i)}}$, by the union bound, we have \eqref{eq:Bound_(III)_for_one_pair} for all $\hat{y}$ in $\mathcal{N}_{K^{(i)}}$ with probability greater than
\begin{align*}
&1 - 2 |  \mathcal{N}_{K^{(i)}} | \exp \left(-c'  (600C)^{2 \left(1 -\left({5 \over 6} \right)^{i-1} \right)} { m^{ {1 \over 3}\left({5 \over 6} \right)^{i-1} } (\log m)^{ {14 \over 3}\left(1 -\left({5 \over 6} \right)^{i-1} \right) +{4 \over 3}}  C(N,s,m)^{5  - {10 \over 3}\left({5 \over 6} \right)^{i-1}}\over 60}  \right)\\
& \ge 1 -2 \exp \left( C^{''} (600C)^{2 \left(1 -\left({5 \over 6} \right)^{i-1} \right)} m^{{1 \over 3}({5 \over 6})^{i-1}} (\log m)^{2 \left( 1- {7 \over 3}\left({5 \over 6} \right)^{i-1} \right)} C(N,s,m)^{4 - {10 \over 3}\left({5 \over 6} \right)^{i-1}}  \right) \\
& \qquad \qquad  \times  \exp \left(-c'  (600C)^{3 \left(1 -\left({5 \over 6} \right)^{i-1} \right)} { m^{ {1 \over 3}\left({5 \over 6} \right)^{i-1} }  (\log m)^{ {14 \over 3}\left(1 -\left({5 \over 6} \right)^{i-1} \right) +{4 \over 3}}  C(N,s,m)^{5  - {10 \over 3}\left({5 \over 6} \right)^{i-1}}\over 60}  \right).
\end{align*}
By the same argument in the previous subsection based on comparing the exponents of $\log m$ and $C(N,s,m)$, this probability is at least 
\[
1 - \exp \left(-C^{(5)} m^{ {1 \over 3}\left({5 \over 6} \right)^{i} } (\log m)^{7 \over 6} C(N,s,m) \right) \ge 1 - {c_7 \over m^5}.
\]

\paragraph{Step 3. Establishing the bound for $(III)$ } 
Combining the previous bounds in this subsection yields
\begin{align*}
\left \| {\tau \over m} \sum_{i=1}^m(\text{sign} (\inner{a_i,x}) - \text{sign}(\inner{a_i, y}) b_i(x,y) \right \|_{K_1^\circ}  
&\le {4 \tau  \over 600}r_{i+1} + {2 \tau r_{i+1} \over 600} +  {2 \tau r_{i+1} \over 600}\\
&\le  \left( {6 \over 600} +  {3 \over 600} +  {3 \over 600} \right)r_{i+1} \\
&=  {r_{i+1} \over 50} 
\end{align*} with probability exceeding
\begin{align*}
& 1 - {2 \over m^4} - {c_2 \over m^5} - {2 \over m^4} - {c_7 \over m^4} - \bar{c} \left({s \over N} \right)^{5 s} \cdot {1 \over m^5}  \ge 1 -   \bar{c}  \left({s \over N} \right)^{5 s} \cdot {1 \over m^5} - {c_8 \over m^4}.
\end{align*}

\section{Discussion}
We show that NBHIT enjoys the optimal approximation error decay in the number of measurements for the one-bit compressed sensing problem. While this demonstrates its efficiency, there are still several aspects worth further investigation: 

(1) 
It would be interesting to see how much the requirement for number of measurements in Theorem \ref{thm: Main Convergence Thm} can be relaxed. Although improving this requirement is not the main focus in this paper, the extensive numerical experiments in the literature indicate that BIHT-type algorithms perform much better than other algorithms even for a moderate number of measurements \cite{jacques2013robust, boufounos2015quantization, liu2019one}. We leave it as a future work. 
 
(2) Note that in general it is not possible to recover a sparse signal from $1$-bit measurements with non-Gaussian vectors even if we have infinitely many measurements \cite{ai2014one}. However, under some extra assumptions on the signal set, we can reconstruct the signal with a reasonable accuracy. 
It could be worth to explore whether BIHT-type algorithms still exhibit superior performance for non-Gaussian measurements under these assumptions. 
%This line of research appears to be relatively less explored. 
%Our analysis relies on the Gaussian property of the measurement vectors, but there are technical tools available to convert the theoretical guarantees in one-bit compressed sensing for Gaussian case to a more general class of measurements \cite{ai2014one,dirksen2017one}. One may be able to generalize our approach through these methods. 

(3) Another possible direction would be to extend and analyze  BIHT-type algorithms for sparse signals with respect to a dictionary. It is easy to see that our results naturally extend to sparse signals with respect to any orthogonal basis. 
We expect that BIHT-type algorithms might offer a good approximation error decay for a certain type of dictionaries as well, assuming that the hard thresholding operator for the dictionary is well defined and can be implemented in a computationally-efficient way.

\section*{Acknowledgement}
The authors thank Xiaowei Li for reading this manuscript and giving us several valuable comments.

\ifCLASSOPTIONcaptionsoff
  \newpage
\fi



%\bibliographystyle{IEEEtran}

%\bibliographystyle{plain}
\bibliographystyle{IEEEtran}
\bibliography{ApproximateRIP_BIHT_Ref}

\begin{thebibliography}{10}
\providecommand{\url}[1]{#1}
\csname url@samestyle\endcsname
\providecommand{\newblock}{\relax}
\providecommand{\bibinfo}[2]{#2}
\providecommand{\BIBentrySTDinterwordspacing}{\spaceskip=0pt\relax}
\providecommand{\BIBentryALTinterwordstretchfactor}{4}
\providecommand{\BIBentryALTinterwordspacing}{\spaceskip=\fontdimen2\font plus
\BIBentryALTinterwordstretchfactor\fontdimen3\font minus
  \fontdimen4\font\relax}
\providecommand{\BIBforeignlanguage}[2]{{%
\expandafter\ifx\csname l@#1\endcsname\relax
\typeout{** WARNING: IEEEtran.bst: No hyphenation pattern has been}%
\typeout{** loaded for the language `#1'. Using the pattern for}%
\typeout{** the default language instead.}%
\else
\language=\csname l@#1\endcsname
\fi
#2}}
\providecommand{\BIBdecl}{\relax}
\BIBdecl

\bibitem{foucart2013invitation}
S.~Foucart and H.~Rauhut, ``An invitation to compressive sensing,'' in \emph{A
  mathematical introduction to compressive sensing}.\hskip 1em plus 0.5em minus
  0.4em\relax Springer, 2013, pp. 1--39.

\bibitem{eldar2012compressed}
Y.~C. Eldar and G.~Kutyniok, \emph{Compressed sensing: theory and
  applications}.\hskip 1em plus 0.5em minus 0.4em\relax Cambridge university
  press, 2012.

\bibitem{gunturk2013sobolev}
C.~S. G{\"u}nt{\"u}rk, M.~Lammers, A.~M. Powell, R.~Saab, and {\"O}.~Y{\i}lmaz,
  ``Sobolev duals for random frames and {$\Sigma \Delta$} quantization of
  compressed sensing measurements,'' \emph{Foundations of Computational
  mathematics}, vol.~13, no.~1, pp. 1--36, 2013.

\bibitem{saab2018quantization}
R.~Saab, R.~Wang, and {\"O}.~Y{\i}lmaz, ``Quantization of compressive samples
  with stable and robust recovery,'' \emph{Applied and Computational Harmonic
  Analysis}, vol.~44, no.~1, pp. 123--143, 2018.

\bibitem{chou2015noise}
E.~Chou, C.~S. G{\"u}nt{\"u}rk, F.~Krahmer, R.~Saab, and {\"O}.~Y{\i}lmaz,
  ``Noise-shaping quantization methods for frame-based and compressive sampling
  systems,'' in \emph{Sampling theory, a renaissance}.\hskip 1em plus 0.5em
  minus 0.4em\relax Springer, 2015, pp. 157--184.

\bibitem{chou2016distributed}
E.~Chou and C.~S. G{\"u}nt{\"u}rk, ``Distributed noise-shaping quantization: I.
  beta duals of finite frames and near-optimal quantization of random
  measurements,'' \emph{Constructive Approximation}, vol.~44, no.~1, pp. 1--22,
  2016.

\bibitem{baraniuk2017exponential}
R.~G. Baraniuk, S.~Foucart, D.~Needell, Y.~Plan, and M.~Wootters, ``Exponential
  decay of reconstruction error from binary measurements of sparse signals,''
  \emph{IEEE Transactions on Information Theory}, vol.~63, no.~6, pp.
  3368--3385, 2017.

\bibitem{deift2011optimal}
P.~Deift, F.~Krahmer, and C.~S. G{\"u}nt{\"u}rk, ``An optimal family of
  exponentially accurate one-bit sigma-delta quantization schemes,''
  \emph{Communications on Pure and Applied Mathematics}, vol.~64, no.~7, pp.
  883--919, 2011.

\bibitem{saab2017compressed}
R.~Saab, R.~Wang, and {\"O}.~Y{\i}lmaz, ``From compressed sensing to compressed
  bit-streams: practical encoders, tractable decoders,'' \emph{IEEE
  Transactions on Information Theory}, vol.~64, no.~9, pp. 6098--6114, 2017.

\bibitem{boufounos20081}
P.~T. Boufounos and R.~G. Baraniuk, ``1-bit compressive sensing,'' in
  \emph{2008 42nd Annual Conference on Information Sciences and Systems}.\hskip
  1em plus 0.5em minus 0.4em\relax IEEE, 2008, pp. 16--21.

\bibitem{plan2016high}
Y.~Plan, R.~Vershynin, and E.~Yudovina, ``High-dimensional estimation with
  geometric constraints,'' \emph{Information and Inference: A Journal of the
  IMA}, vol.~6, no.~1, pp. 1--40, 2016.

\bibitem{plan2016generalized}
Y.~Plan and R.~Vershynin, ``The generalized lasso with non-linear
  observations,'' \emph{IEEE Transactions on information theory}, vol.~62,
  no.~3, pp. 1528--1537, 2016.

\bibitem{davenport20141}
M.~A. Davenport, Y.~Plan, E.~Van Den~Berg, and M.~Wootters, ``1-bit matrix
  completion,'' \emph{Information and Inference: A Journal of the IMA}, vol.~3,
  no.~3, pp. 189--223, 2014.

\bibitem{horowitz2009semiparametric}
J.~L. Horowitz, \emph{Semiparametric and nonparametric methods in
  econometrics}.\hskip 1em plus 0.5em minus 0.4em\relax Springer, 2009,
  vol.~12.

\bibitem{jacques2013robust}
L.~Jacques, J.~N. Laska, P.~T. Boufounos, and R.~G. Baraniuk, ``Robust 1-bit
  compressive sensing via binary stable embeddings of sparse vectors,''
  \emph{IEEE Transactions on Information Theory}, vol.~59, no.~4, pp.
  2082--2102, 2013.

\bibitem{boufounos2015quantization}
P.~T. Boufounos, L.~Jacques, F.~Krahmer, and R.~Saab, ``Quantization and
  compressive sensing,'' in \emph{Compressed sensing and its
  applications}.\hskip 1em plus 0.5em minus 0.4em\relax Springer, 2015, pp.
  193--237.

\bibitem{liu2019one}
D.~Liu, S.~Li, and Y.~Shen, ``One-bit compressive sensing with projected
  subgradient method under sparsity constraints,'' \emph{IEEE Transactions on
  Information Theory}, vol.~65, no.~10, pp. 6650--6663, 2019.

\bibitem{jacques2013quantized}
L.~Jacques, K.~Degraux, and C.~De~Vleeschouwer, ``Quantized iterative hard
  thresholding: Bridging 1-bit and high-resolution quantized compressed
  sensing,'' \emph{arXiv preprint arXiv:1305.1786}, 2013.

\bibitem{powell2013quantization}
A.~M. Powell, R.~Saab, and {\"O}.~Y{\i}lmaz, ``Quantization and finite
  frames,'' in \emph{Finite frames}.\hskip 1em plus 0.5em minus 0.4em\relax
  Springer, 2013, pp. 267--302.

\bibitem{plan2013one}
Y.~Plan and R.~Vershynin, ``One-bit compressed sensing by linear programming,''
  \emph{Communications on Pure and Applied Mathematics}, vol.~66, no.~8, pp.
  1275--1297, 2013.

\bibitem{plan2012robust}
------, ``Robust 1-bit compressed sensing and sparse logistic regression: A
  convex programming approach,'' \emph{IEEE Transactions on Information
  Theory}, vol.~59, no.~1, pp. 482--494, 2012.

\bibitem{knudson2016one}
K.~Knudson, R.~Saab, and R.~Ward, ``One-bit compressive sensing with norm
  estimation,'' \emph{IEEE Transactions on Information Theory}, vol.~62, no.~5,
  pp. 2748--2758, 2016.

\bibitem{rangan2001recursive}
S.~Rangan and V.~K. Goyal, ``Recursive consistent estimation with bounded
  noise,'' \emph{IEEE Transactions on Information Theory}, vol.~47, no.~1, pp.
  457--464, 2001.

\bibitem{nguyen2011frame}
H.~Q. Nguyen, V.~K. Goyal, and L.~R. Varshney, ``Frame permutation
  quantization,'' \emph{Applied and Computational Harmonic Analysis}, vol.~31,
  no.~1, pp. 74--97, 2011.

\bibitem{powell2010mean}
A.~M. Powell, ``Mean squared error bounds for the rangan--goyal soft
  thresholding algorithm,'' \emph{Applied and Computational Harmonic Analysis},
  vol.~29, no.~3, pp. 251--271, 2010.

\bibitem{oymak2015near}
S.~Oymak and B.~Recht, ``Near-optimal bounds for binary embeddings of arbitrary
  sets,'' \emph{arXiv preprint arXiv:1512.04433}, 2015.

\bibitem{ledoux2001concentration}
M.~Ledoux, \emph{The concentration of measure phenomenon}.\hskip 1em plus 0.5em
  minus 0.4em\relax American Mathematical Soc., 2001, no.~89.

\bibitem{vershynin2018high}
R.~Vershynin, \emph{High-dimensional probability: An introduction with
  applications in data science}.\hskip 1em plus 0.5em minus 0.4em\relax
  Cambridge University Press, 2018, vol.~47.

\bibitem{bilyk2015random}
D.~Bilyk and M.~T. Lacey, ``Random tessellations, restricted isometric
  embeddings, and one bit sensing,'' \emph{arXiv preprint arXiv:1512.06697},
  2015.

\bibitem{chen2017solving}
Y.~Chen and E.~J. Cand{\`e}s, ``Solving random quadratic systems of equations
  is nearly as easy as solving linear systems,'' \emph{Communications on Pure
  and Applied Mathematics}, vol.~70, no.~5, pp. 822--883, 2017.

\bibitem{ai2014one}
A.~Ai, A.~Lapanowski, Y.~Plan, and R.~Vershynin, ``One-bit compressed sensing
  with non-gaussian measurements,'' \emph{Linear Algebra and its Applications},
  vol. 441, pp. 222--239, 2014.

\end{thebibliography}

\end{document}